\documentclass[11pt]{article}
\usepackage{setspace,graphicx}
\usepackage[margin=1in]{geometry}

\usepackage{times}
\usepackage{mathtools}
\usepackage{amssymb,amsmath}
\usepackage{amsthm}

\usepackage{enumitem}

\usepackage[linesnumbered, vlined, ruled]{algorithm2e}
\usepackage{tikz}
\usetikzlibrary{decorations.pathreplacing}
\AtBeginDocument{%
  }

\begin{document}

\title{Ultra-Resilient Superimposed Codes:\\ Near-Optimal Construction and Applications}

\newtheorem{lemma}{Lemma}
\newtheorem{theorem}{Theorem}
\newtheorem{esempio}{Example}
\newtheorem{corollary}{Corollary}
\newtheorem{proposition}{Proposition}
\newtheorem{example}{Example}[section] 
\newtheorem{definition}{Definition}
\newtheorem{claim}{Claim}

\definecolor{apricot}{rgb}{0.98, 0.81, 0.69}
\definecolor{beige}{rgb}{0.96, 0.96, 0.86}
\definecolor{blizzardblue}{rgb}{0.67, 0.9, 0.93}
\definecolor{non-photoblue}{rgb}{0.64, 0.87, 0.93}
\definecolor{celadon}{rgb}{0.67, 0.88, 0.69}
\definecolor{grannysmithapple}{rgb}{0.66, 0.89, 0.63}
\definecolor{macaroniandcheese}{rgb}{1.0, 0.74, 0.53}

\newcommand{\remove}[1]{}
\newcommand{\x}{{\bf x}}
\newcommand{\bc}{{\mathbf  c}}
\newcommand{\bd}{{\mathbf  d}}
\newcommand{\be}{{\mathbf  e}}
\newcommand{\bg}{{\mathbf  g}}
\newcommand{\bx}{{\mathbf  x}}
\newcommand{\by}{{\mathbf  y}}
\newcommand{\bb}{{\mathbf  b}}
\newcommand{\bz}{{\mathbf  z}}
\newcommand{\M}{{\mathbf  M}}
\newcommand{\pairCondition}{$\mathcal{P} (\mathbf{M}, \alpha, \tau_1, \tau_2)$}
\newcommand{\elongation}{elongation}

\newcommand{\B}{\vspace*{-2.5ex}}

\newcommand{\cS}{{\cal S}}
\newcommand{\cA}{{\cal A}}
\newcommand{\cM}{{\mathcal{M}}}

\newcommand{\gdm}[1]{{\color{black}{#1}}}

\newcommand{\dk}[1]{{\color{black}{#1}}}
\newcommand{\gia}[1]{{\color{black}{#1}}}

\newcommand{\uv}[1]{{\color{black}{#1}}}

\newtheorem{fact}{Fact}

\makeatletter

\newcommand{\newreptheorem}[2]{%
\newenvironment{rep#1}[1]{%
 \def\rep@title{#2 \ref{##1}}%
 \begin{rep@theorem}}%
 {\end{rep@theorem}}}
\makeatother

%\newtheorem*{theorem*}{Theorem}

%\begin{CCSXML}
%<ccs2012>
%   <concept>
%       <concept_id>10003752.10003809.10010172</concept_id>
%       <concept_desc>Theory of computation~Distributed algorithms</concept_desc>
%       <concept_significance>500</concept_significance>
%       </concept>
%\end{CCSXML}

%\ccsdesc[500]{Theory of computation~Distributed algorithms}

\date{}


\author{
	Gianluca De Marco\thanks{University of Salerno, Salerno, Italy
		({gidemarco@unisa.it})}
	\and Dariusz R. Kowalski\thanks{Augusta University, Augusta, Georgia, USA
		({dkowalski@augusta.edu })}
}

\maketitle

%\vspace*{-6ex}
\begin{abstract}
A superimposed code is a collection of binary vectors (codewords) with the property that no vector is contained 
in the Boolean sum of any $k$ others, enabling unique identification of codewords within any group of 
$k$. 
Superimposed codes are foundational combinatorial tools 
with applications in areas ranging from 
distributed computing and data retrieval to fault-tolerant communication. However, classical superimposed 
codes rely on strict alignment assumptions, limiting their effectiveness in asynchronous and fault-prone 
environments, which are common in modern systems and applications.

We introduce Ultra-Resilient Superimposed Codes (URSCs), a new class of codes that extends 
the classic superimposed framework by ensuring a stronger codewords' isolation property and resilience to two types of adversarial perturbations: arbitrary cyclic shifts and partial bitwise corruption (flips).
Additionally, URSCs exhibit universality, adapting seamlessly to any number 
$k$ of concurrent codewords without prior knowledge. 
This is a combination of properties not achieved in any previous construction.

We provide the first polynomial-time construction of URSCs with near-optimal length, significantly 
outperforming previous constructions with less general features, all without requiring prior knowledge 
of the number of concurrent codewords, $k$.
We demonstrate that our URSCs significantly advance the state of the art in multiple 
applications, 
including uncoordinated beeping networks, where our codes reduce time complexity for local broadcast 
by nearly two orders of magnitude, and generalized contention resolution in multi-access channel communication. 
%
%Our results open new avenues for deploying superimposed codes in asynchronous and adversarial environments, 
%enhancing robustness and scalability in critical systems.

\end{abstract}

%\keywords{algorithms, randomized algorithms, superimposed codes}

%\thispagestyle{empty}
%\setcounter{page}{0}

%\newpage

\section{Introduction}

\textit{Superimposed codes}, proposed by Kautz and Singleton in 1964 \cite{KS1964}, are a well-known and widely used class 
of codes represented by a binary $t \times n$ matrix, where columns correspond to codewords. 
These codes have a distinctive property: for any subset of  
$k$ columns from the matrix, and for any column $\mathbf{c}$ within this subset, there exists at least one row where 
column $\mathbf{c}$ has an entry of 1 while all other columns in the subset have entries of 0. 
This feature allows for the unique identification of any single column (codeword) within a subset of columns.

The significance of superimposed codes cannot be overstated. 
Over the past six decades, these codes have found applications in an impressively wide range of fields. 
For instance, they have proven instrumental in information retrieval (see e.g. \cite[p.~570]{K1998}), 
pattern matching \cite{I1997, PR11}, learning problems \cite{AA2005, KesselheimKN15},
wireless and multi-access communication \cite{J1985, I2002, Chrobak2000}, 
distributed coloring \cite{N1992, RG2010}.

Despite their broad applicability, classic superimposed codes have a significant limitation: they require 
perfectly aligned (i.e., fixed) codewords. For example, in wireless distributed communication, this alignment 
requires synchronization before the codes can be applied. Similarly, in distributed information retrieval, 
it necessitates the alignment of file descriptor codewords. However, achieving synchronization is often challenging 
or impractical in many distributed 
environments, where entities may start using codewords at arbitrary 
times as they wake up and join the system. Dynamic data retrieval also presents difficulties, as it may 
involve data arriving out of order. The lack of synchronization creates a difficulty in defining a meaningful product of two codes. Other limitations arise from a limited capacity to handle data~corruption.
Finally, existing codes assume that the number of superimposed codewords, 
$k$, is known. However, some applications may not provide this in advance, and standard estimation techniques—such as doubling the estimate of 
$k$—are ineffective, as consecutive code estimates can overlap due to arbitrary misalignments~of~codewords.

% OLDER: To address these limitations, more general and flexible codes have recently been proposed. However, existing literature either focuses on perfectly aligned codewords (thereby restricting applications to synchronized systems) or introduces less efficient codes than their classic counterparts, also assuming knowledge of the number $k$ of superimposed codewords as an input parameter. This assumption significantly limits the applicability of these codes to scenarios where the number of contending entities (such as stations in a distributed system) is known in advance.

The purpose of the present work is to introduce and efficiently construct new superimposed codes that 
simultaneously address all the above limitations. 
We provide a near-optimal solution that significantly 
outperforms prior codes, which were designed to ensure only some or weaker properties, and demonstrate usefulness of new code's properties in various applications.

\subsection{Our contributions and comparison with previous works}

Before presenting the formal description in Section \ref{sec:preliminaries} (Definition~\ref{def:shift_k_unknown-fuzzy}), we will provide a 
high-level overview of our codes. 
We begin with the classic definition of superimposed codes, followed by an outline of the properties of our codes, including both the newly introduced and previously established 
properties.
%and the technical innovations that enabled us to tackle the main challenges. 
We will then review the relevant literature related to our codes and compare their performance to prior 
results. Finally, we will conclude with a discussion on how these codes can enhance performance in pertinent 
distributed problems, particularly in beeping networks and multiple-access channels, as well as  
in other contexts.

\begin{definition}[Classic definition]\label{classic}
A $(k, n)$-superimposed code is formally defined as a $t \times n$ binary matrix $\mathbf{M}$
such that the following property holds:
\vspace*{-1ex}
\begin{itemize}[leftmargin=5.5mm]
    \item[] (${\mathcal I}$) For any $k$-tuple of columns of $\mathbf{M}$  
and for any column $\mathbf{c}$ of the given $k$-tuple, there is a row $0 \le i < t$ 
such that column $\mathbf{c}$ has $1$ in row $i$ and all remaining $k -1$ columns of the $k$-tuple have
all $0$ in row $i$. 
\end{itemize}
\vspace*{-1ex}
The columns of matrix $\mathbf{M}$ are called {\em codewords} and their length $t$ is called the 
{\em code length}.   
\end{definition}

%\noindent
%Property ${\mathcal I}$ is equivalent to: the bitwise OR of any $k$ codewords does not contain any other codeword.

\subsubsection{Ultra-Resilient Superimposed Codes (URSC)}

Our newly proposed ultra-resilient superimposed codes represent a substantial generalization of the classic definition, designed to provide robust performance without requiring prior knowledge of the number $k$ of superimposed codes.
%The formal specification is given in Section~\ref{sec:preliminaries}, see Definition~\ref{def:shift_k_unknown-fuzzy}.
We start from the following enhanced version of property ${\mathcal I}$,
which ensures a better isolation of codewords; consequently, as shown in Section~\ref{sec:beeping} (Algorithm~\ref{alg:beeping} and Lemma~\ref{lem:beeping-progress}), it allows to implement a meaningful product of codes even in the presence of arbitrary
codeword misalignment~and~bit~corruption.
\vspace*{-1ex}
\begin{itemize}[leftmargin=5.5mm]
    \item[] \textbf{Isolation property:} (${\mathcal I'}$) For any $k$-tuple of columns of $\mathbf{M}$ and any column $\mathbf{c}$ 
    of this $k$-tuple, there exists a row $0 \le i < t$ such that column $\mathbf{c}$ has a $1$ in row $i$, 
    while all other $k-1$ columns of the $k$-tuple have $0$ in rows 
    $(i-1)\bmod t$, $i$, and $(i+1)\bmod t$.
\end{itemize}

\vspace*{-1ex}
\noindent
At a high level, the ultra-resilient capability of our codes simultaneously ensures the 
following:

\vspace*{-1ex}
\begin{itemize} 
    \item[\textbf{(a)}] \textbf{Shift resilience:} Even if each codeword undergoes independent and arbitrary cyclic shifts, property ${\mathcal I'}$ is still preserved. This resilience is particularly effective in uncoordinated settings where computational entities can join the system at arbitrary times (asynchronous activations).

\vspace*{-1ex}
    \item[\textbf{(b)}] \textbf{Flip resilience:} Property ${\mathcal I'}$ remains intact 
    even if a fraction $1-\alpha$ of unique occurrences of ones in each codeword can be flipped 
    %\textit{for any arbitrary cyclic shift of the codewords} 
    (the classic setting is obtained for $\alpha = 1$).
    This ensures robust performance under adversarial jamming in dynamic, asynchronous environments. 
%    \item[\textbf{(b)}] \textbf{Flip resilience:} Property ${\mathcal I'}$ remains intact even if a fraction $1-\alpha$ of unique occurrences of ones in each codeword can be flipped (the classic setting is obtained for $\alpha = 1$). This ensures robust performance even in the presence of adversarial jamming.
    
\vspace*{-1ex}
    \item[\textbf{(c)}] \textbf{Universality:} The construction and application of these codes do not require prior knowledge of the parameter $k$, the number of concurrently active codewords. This intrinsic feature enables their applicability across dynamic and uncertain settings where $k$ may vary or remain unknown.
\end{itemize}

\vspace*{-1ex}
\noindent
%\textcolor{red}{
\textbf{Note:} Our codes ensure that flip resilience is fully integrated with shift resilience, 
so that these properties work together to reinforce robustness. Specifically, flip resilience 
is maintained under arbitrary cyclic shifts, meaning that for any shifted configuration, 
the isolation property $\mathcal{I'}$ holds even if a fraction $1 - \alpha$ of ones in each codeword is 
flipped. This combined ultra-resilience provides robust protection in dynamic and adversarial environments, 
where both positional misalignments and bitwise corruption may occur simultaneously.
%}

\vspace*{1ex}
%\medskip
In a classic superimposed code (where in particular shift resilience is not required)
the isolation property ${\mathcal I'}$ can be ensured by interleaving each row of its matrix 
$\mathbf{M}$ with a row of 0's. 
Specifically, for each row $i$ of $\mathbf{M}$, the interleaved matrix $\mathbf{M'}$ has
row $2i$ as row $i$ of $\mathbf{M}$ and row $2i+1$ filled entirely with 0's.
However, in our ultra-resilient superimposed codes, where shift resilience has to be ensured,
this technique fails as the independent shifts of columns would disrupt this row-by-row alignment,
%under arbitrary shifts, 1-bits may be moved into positions that 
%lack the necessary isolation, 
resulting in overlaps that violate~${\mathcal I'}$. 
Thus, achieving the isolation property with shift resilience requires a much more robust construction.

\medskip
In classic superimposed codes, the length $t$ of the codewords (Definition \ref{classic})
serves as a key measure of efficiency and performance.
However, in the more complex scenario we address—where universality must be maintained in the presence of arbitrary misalignments of the codewords—this fixed measure becomes impractical 
(see Section \ref{sec:preliminaries} for a more detailed discussion of this challenge). 
To overcome this limitation, we introduce a more advanced metric, that we call 
\textbf{\em code elongation}, that dynamically adjusts to the actual (but unknown) 
number of superimposed codewords in the system.
%This metric represents the minimum length of the codeword prefix required 
%to guarantee the ultra-resilient properties for the actual (unknown) number of superimposed codewords, $k$.

%To address this challenge, we introduce a more advanced metric, termed 
%\textit{code elongation} which represents the minimum length of the codeword prefix 
%that guarantees the ultra-resilient properties for the actual (unknown) 
%number of superimposed codewords, $k$. 
%In other words, it is a variable length that dynamically adjusts to the actual (but unknown) 
%number of superimposed codewords in the system, providing a flexible and precise way to evaluate 
%the performance of the codes without requiring prior knowledge of $k$.

\B
%\medskip
\paragraph*{Comparative Efficiency and Performance.}
Our work builds on and extends foundational studies on superimposed codes, enhancing them to address the unique 
challenges posed by dynamic and asynchronous communication scenarios, fault tolerance, and the absence of knowledge 
about the number $k$ of superimposed codewords. 
Our main result, stated in Theorem~\ref{thm:code-final} of Section~\ref{sec:unknown-k}, is as follows:
\vspace{-0.5ex}
\begin{quote}    
{\em We provide a Las Vegas algorithm that, in polynomial time, 
constructs an ultra-resilient superimposed code with a near-optimal performance (code elongation) 
of $O\left(\frac{k^2}{\alpha^{2+\epsilon}} \log n\right)$.}
\end{quote}

\vspace{-0.5ex}
\noindent
%\dk{[[[I WOULD COMMENT OUT THIS PARAGRAPH]]]}
%To highlight the impact and effectiveness of our results, we contextualize our contributions within the 
%historical framework of superimposed codes, drawing comparisons to existing results. Rather than being 
%exhaustive, the following overview will focus on the key results that are directly related to our contribution.

For classic superimposed codes (i.e., those with aligned codewords and no flip resilience), the 
best-known existential bound is $O(k^2 \log(n/k))$, established by Erd{\"o}s, Frankl, 
and F{\"u}redi~\cite{EFF1982}. 
Over the years, various proofs have been presented for the corresponding lower bound, including those by D'yachkov and Rykov \cite{AV1982}, Ruszinkó \cite{R1994}, and Füredi \cite{F1996}, 
all converging on an $\Omega(k^2 \log_k n)$ bound. An elegant combinatorial proof of this lower bound 
has also been provided more recently by Alon and Asodi \cite{AA2005}. Porat and Rothschild~\cite{PR11} 
made a significant contribution by introducing the first polynomial-time algorithm to construct classic 
superimposed codes with a near-optimal length of $O(k^2 \log n)$.

Recently, Rescigno and Vaccaro \cite{RV2024} introduced a generalized fault-tolerant version of 
classic superimposed codes ({with flip resilience, but without shift resilience}). 
They presented a randomized construction that achieves an 
average code length of $O((\frac{k}{\alpha})^2 \log n)$ in polynomial time.
%, where $1 - \alpha$ representing the fraction of unique occurrences of ones that can be flipped while still 
%preserving the property defined in Definition \ref{classic}. 
Additionally, they demonstrated that 
any superimposed code supporting flip resilience requires a length of $\Omega\Big(\Big(\frac{k}
{\alpha}\Big)^2 \frac{\log n}{\log (\frac{k}{\alpha})}\Big)$.

An early extension addressing non-aligned codewords for use in synchronization problems was introduced by 
Chu, Colbourn, and Syrotiuk \cite{CCSb2006,CCSa2006}, who proposed a generalized version of 
superimposed codes with cyclic shifts and fault tolerance. In terms of our parameters, their construction 
achieves an efficiency of $O((k \frac{\log n}{\log k})^3)$, assuming a fixed constant 
$\alpha$. Another significant contribution in the pursuit of generalizing classic superimposed codes to 
accommodate non-aligned codewords was achieved recently by Dufoulon, Burman, and Beauquier \cite{DBB2020} 
in the context of asynchronous beeping models. They introduced polynomial-time constructible superimposed 
codes, called Uncoordinated Superimposed Codes (USI-codes), which effectively handle arbitrary shifts 
of codewords (without fault tolerance), exhibiting a code length of $O(k^2 n^2)$.

It is important to note that all of the above upper bounds were achieved with knowledge of $k$, which 
played a significant supportive role in the design of the respective algorithms. Notably, the lack of 
knowledge about $k$ presents a challenge for construction, only when shifts are introduced. 
Otherwise, it suffices to construct codes tailored to fixed values of $k$, as is typically assumed 
in the literature (cf.~\cite{PR11}), and concatenate them for exponentially growing values of $k$.

%\medskip
With a performance of $O\left(\frac{k^2}{\alpha^{2+\epsilon}} \log n\right)$, our codes significantly 
outperform all prior results for non-aligned codewords, and they do so in a much more general and 
challenging dynamic setting. 
{Specifically, our codes provide 
\textit{a stronger isolation property ${\mathcal I'}$ that accommodates shift resilience, 
flip resilience, and universality}. This means that each codeword retains isolated 1-bits in a local 
neighborhood of three adjacent rows, even under arbitrary cyclic shifts, up to a $1-\alpha$ fraction 
of bit flips for any 
 $0 < \alpha \le 1$, and without prior knowledge of the parameter $k$, the number of active codewords. 
This combination of properties is unprecedented in the field.
}

Additionally, our codes are the \textit{first} to achieve \textit{near-optimal} 
performance in the generalization of both shift resilience and flip resilience (fault tolerance),
closely adhering to the lower bound~established for \textit{fixed, aligned} codewords~\cite{RV2024}.
{Importantly, all prior constructions combining shift and flip resilience performed substantially worse, even without guaranteeing the isolation property or universality.}

If we set aside fault tolerance (i.e., by considering $\alpha = 1$), our \textit{polynomial-time constructible} 
codes achieve a {code elongation} of $O(k^2 \log n)$, which matches the best existential bound 
for \textit{classic}
superimposed codes by Erd{\"o}s, Frank and F{\"u}redi~\cite{EFF1982} for all $k = n^{o(1)}$. 
Notably, our codes achieve this same bound while additionally exhibiting 
{our isolation property with shift and flip resilience and without 
requiring prior knowledge of the parameter $k$}. Moreover, it is important to note that we also almost
match the fundamental lower bound $\Omega(k^2 \log_k n)$, as established in 
\cite{AV1982, R1994, F1996, AA2005}, 
which is valid even for classic superimposed codes (Definition \ref{classic})). 
Also, our construction yields codes with asymptotically the same length $O(k^2 \log n)$ as the best-known 
polynomially constructible classic superimposed codes \cite{PR11}, \textit{i.e.}, codes requiring codewords 
to be aligned and without fault tolerance ($\alpha = 1$).

\iffalse
In Section~\ref{sec:applications} we show how to apply the ultra-resilient superimposed codes to solve 
contention resolution problem on a multiple-access channel.
The best deterministic protocol so far, was existential -- there are transmission sequences guaranteeing 
latency $O(k^2\log n)$ of each station arriving at arbitrary time, 
where $k$ is the maximum number of stations joining the channel~\cite{DEMARCO20231}. 
Using our polynomially constructible codes, we show that now we can efficiently solve contention 
resolution problem with the same asymptotic latency.
We also argue how to improve several other distributed problems, in which synchronization 
(that would have allowed using e.g., classic superimposed codes) could be costly.
\fi

\B
\paragraph*{Technical novelty.}
Rooted in the universality concept of code elongation, our construction is efficiently achieved through a novel 
de-randomization of a specifically designed stochastic matrix distribution 
(as defined in Definition~\ref{randomatrix}). We prove that this distribution satisfies 
a crucial property that we term the Collision Bound Property (Definition \ref{CBproperty}) 
with high probability (cf. Lemma \ref{lep}). This property is essential for translating conditions 
on subsets of $k$ columns into conditions that apply to pairs of columns, allowing us to verify and 
construct the ultra-resilient superimposed codes in a computationally efficient manner (cf. Lemma \ref{1e3}).
In this way, we dramatically reduce the size of the problem space from 
considering all $k$-subsets of columns to focusing on pairs of columns and their possible shifts, see Algorithm~\ref{alg:k_unknown}. 
This transformation reduces the complexity of the problem to polynomial time in terms of the number of codewords $n$, making the de-randomization feasible and scalable. 
Once the problem space is reduced to pairs of columns, we must ensure that all desired code properties (as outlined in Definition~\ref{CBproperty}) are preserved during the de-randomization process. This involves analyzing intervals within codewords and accounting for three types of shifts and rearrangements. The proof of Lemma~\ref{lep} addresses this in detail; see Section~\ref{satisfaction} for an overview of the challenges and methods involved, and Appendix~\ref{proofs} for the complete~proof.

\subsubsection{Applications}
%It is important to note that 

{
The two main applications studied in this paper are within the contexts of Beeping networks and Contention resolution, which are outlined below. 
We direct the reader to Appendix \ref{sec:final} for a discussion of many other 
potential applications.
}
It is important to note that, although the code is obtained by a Las Vegas randomized algorithm, 
the resulting codewords are fixed (i.e., non-probabilistic) and can be used reliably in deterministic algorithms. Furthermore, the randomized 
algorithm only needs to be run once, and the generated code can be reused 
as many times as needed.

\B
\paragraph*{Deterministic Neighborhood learning and Local broadcast in beeping networks (Section~\ref{sec:beeping}).}
The beeping model, introduced by Cornejo and Kuhn \cite{CornejoK10}, 
is a minimalist communication framework where nodes communicate in discrete time slots by 
either beeping or remaining silent, and the only feedback a node receives is whether there 
was a beep in the current time slot.

In the context of uncoordinated beeping networks, our novel coding tool significantly improves the time 
complexity for the local broadcast problem, previously addressed in \cite{DBB2020}, and related neighborhood learning. The 
%earlier 
solution in \cite{DBB2020} achieved a time complexity of $O(\Delta^4 M)$, where $\Delta$ is the maximum node degree and 
$M$ is the message size. 

\vspace*{-1ex}
\begin{quote}
{\em Our approach (see Section~\ref{sec:beeping} and Theorem~\ref{thm:beeping}) nearly quadratically reduces the complexity of local broadcast to a deterministic 
$O(\Delta^2 \log n \cdot (M + \log n))$, where $n$ is the number~of~nodes.} 
\end{quote}

\vspace*{-1ex}
\noindent
This improvement is made 
possible through our ultra-resilient superimposed codes, which enable more efficient and resilient 
data transmission even under adversarial jamming. More specifically, shift resilience mitigates the impact of uncoordinated activation, while isolation allows to use a product of the code with specifically designed small pieces of code that carry desired information despite of shifts of the main code.
%\dk{??? DO WE NEED THE NEXT SENTENCE: ???} This substantial enhancement underscores the 
%effectiveness of our codes in advancing the robustness and efficiency of communication within uncoordinated 
%beeping networks.

\B
\paragraph*{Deterministic generalized Contention resolution on multi-access channels (Appendix~\ref{sec:applications}).}
In Contention resolution (CR) on a multi-access channel, $k$ stations out of a total ensemble of 
$n$ may become active,
%at arbitrary times, 
each with a packet that can be transmitted in a single time slot. The objective is to enable each of these~$k$ contending stations to successfully transmit its packet, i.e., 
%ensuring that each station eventually 
to transmit without causing a collision with others in the same time slot.
The earliest theoretical work on contention resolution dates back over 50 years, primarily with the seminal papers by Abramson \cite{Abramson}, Roberts \cite{Roberts}, and
Metcalfe and Boggs \cite{MR1976}. Since then, CR 
%in multi-access channels 
has developed a long and rich history, addressing areas such as communication tasks, scheduling, fault tolerance, security, energy efficiency, game theory, and more. 
However, for deterministic solutions, only recently has the problem 
been studied in the challenging setting of arbitrary activation times,
%of the stations,
%context of asynchronous station start times, 
with the best known {\em existential} upper bound provided 
by De Marco, Kowalski and Stachowiak \cite{M2023}. \dk{We can apply URSC codes with suitable parameters to efficiently solve an even more general CR problem, in which {\bf\em at least $s$} successful transmissions per station are required; in particular, in Appendix~\ref{sec:applications} we prove:}
%\dk{See Appendix~\ref{sec:applications} for details.}
%In Appendix~\ref{sec:applications} we generalize CR by requiring {\bf\em at least $s$} successful transmissions per station and prove~in~Theorem~\ref{thm:codes-vs-CR}~\dk{that URSC codes with certain parameters guarantee} the following:

\vspace*{-1ex}
\begin{quote}
{\em By simultaneously leveraging all three properties of our codes -- shift resilience, flip resilience, and universality --
we solve the~generalized CR problem for any $k\le n$ contenders (with $k$ unknown), 
ensuring that each contender achieves at least $s$ successful transmissions 
within $O((k+\frac{s}{\log n})^2\log n)$ rounds after activation, for any $s\ge 1$.}
\end{quote}

\vspace*{-1ex}
\noindent
Our \textit{constructive} upper bound matches the \textit{existential} bound in \cite{M2023} (case $s=1$ in our generalized result) 
and gets very close to the lower bound of $\Omega(\frac{k^2}{\log k})$ proved in the same paper.

\begin{figure}[t!]
\vspace*{-3.5ex}
   \centering   
    \begin{tikzpicture}[baseline=(current bounding box.north)]
        \matrix [nodes=draw,column sep=1mm, font=\small]
        {
        \node[draw=none,yshift=0.85cm]{};
        \\        
            \node[draw=none,yshift=0.5cm](P0) {\scriptsize{$\mathbf{[c_0]}$}};
            &[-0.3cm] \node[draw=none,yshift=0.5cm](P1) {\scriptsize{$\mathbf{[c_1]}$}}; 
            &[-0.3cm] \node[draw=none,yshift=0.5cm](P2) {\scriptsize{$\mathbf{[c_2]}$}}; 
            &[-0.3cm] \node[draw=none,yshift=0.5cm](P3) {\scriptsize{$\mathbf{[c_3]}$}}; 
            %&[-0.3cm] \node[draw=none] {\scriptsize{$\mathbf{[c_4]}$}}; 
             & \node[draw=none,yshift=0.5cm] {\phantom{xxxxx}}; 
            \\
        
            \node[fill=macaroniandcheese](PC0) {1}; 
            & \node[fill=macaroniandcheese](PC1){1}; 
            & \node[fill=macaroniandcheese](PC2) {0}; 
            & \node[fill=macaroniandcheese](PC3) {1}; 
            %& \node[fill=macaroniandcheese] {1}; 
            \\
        
            \node {0}; 
            & \node{0}; & \node {1}; & \node {1}; %& \node {1}; 
            \\
       
            \node {0}; 
            & \node{0}; & \node {0}; & \node {1}; %& \node {0};  
            \\
                
             \node {1}; 
            & \node {0}; & \node {1}; 
            & \node {0}; %& \node {0}; 
            \\

            \node {1}; 
            & \node{1}; & \node {0}; & \node {0}; %& \node {0};  
             \\

            \node {0}; 
            & \node{0}; & \node {0}; & \node {0}; %& \node {0}; 
            \\
        
            \node(CT1) {1};  
            & \node{1}; & \node {0}; & \node {0}; %& \node {1}; 
            \\  

            \node[draw=none,fill=gray!20!white,minimum width=13.2pt,minimum height=20pt]  {\smash{\raisebox{-40\depth}{$\vdots$}}};  
            &  \node[draw=none,fill=gray!20!white,minimum width=13.2pt,minimum height=20pt]  {\smash{\raisebox{-40\depth}{$\vdots$}}}; 
            & \node[draw=none,fill=gray!20!white,minimum width=13.2pt,minimum height=20pt]  {\smash{\raisebox{-40\depth}{$\vdots$}}}; 
            & \node[draw=none,fill=gray!20!white,minimum width=13.2pt,minimum height=20pt]  {\smash{\raisebox{-40\depth}{$\vdots$}}}; 
            %& \node[draw=none,fill=gray!20!white,minimum width=13.2pt,minimum height=20pt]  {\smash{\raisebox{-40\depth}{$\vdots$}}};  
            \\       
           \node {1};  
            & \node{1}; & \node {0}; & \node {0}; %& \node {1}; 
            \\  

            \node {0};  
            & \node{1}; & \node {0}; & \node {1}; %& \node {0}; 
            \\

            \node[fill=non-photoblue](TAU) {1}; 
            & \node[fill=non-photoblue] {0}; & \node[fill=non-photoblue] {1}; 
            & \node[fill=non-photoblue] {1}; %& \node[fill=non-photoblue] {0}; 
            \\ 
        };
        \node[xshift=-2.0cm,yshift=3.0cm](CT0){\scriptsize{$\tau(n,k)$}};
        \draw[dotted, ->] (CT0) |- (CT1);
        \draw[dotted, ->] (P0) -- (PC0); 
        \draw[dotted, ->] (P1) -- (PC1);
        \draw[dotted, ->] (P2) -- (PC2);
        \draw[dotted, ->] (P3) -- (PC3);
%    \node at (current bounding box.south)[anchor=south,yshift=-.7cm]{$\vdots$};
    \node at (current bounding box.south)[anchor=south,xshift=-.35cm,yshift=-0.7cm]{$(a)$};
   
    \end{tikzpicture}    
    %\hspace{0.5cm}
    \begin{tikzpicture}[baseline=(current bounding box.north)]
        \matrix [nodes=draw,column sep=1mm, font=\small]
        {
        \node[draw=none, yshift=.5cm](C0) {\scriptsize{$[\mathbf{c_0}]$}}; 
        &[-0.15cm] \node[draw=none,yshift=1.0cm](C1) {\scriptsize{$[\mathbf{c_1}](-1)$ }}; 
        &[-0.80cm] \node[draw=none, yshift=.5cm](C2) {\scriptsize{$[\mathbf{c_2}] (2)$}}; 
        &[-0.7cm] \node[draw=none,yshift=1.0cm](C3) {\scriptsize{$[\mathbf{c_3}] (2)$}}; 
        %&[-0.7cm] \node[draw=none,yshift=.5cm](C4) {\scriptsize{$[\mathbf{c_4}] (0)$}}; 
        &[-0.3cm] \node[draw=none,yshift=1.5cm] {\phantom{c...}}; 
        &[-.5cm] \node[draw=none,yshift=1.0cm](C5) {\scriptsize{$\mathbf{z}(-1)$\phantom{x}}};
        &[-0.4cm] \node[draw=none,yshift=0.5cm](C6) {\scriptsize{$\mathbf{z}$}};
        &[-0.2cm] \node[draw=none,yshift=1.0cm](C7) {\scriptsize{$\mathbf{z}(1)$}};    
        &[0.2cm] \node[draw=none,yshift=0.5cm](C10) {\scriptsize{$\mathbf{z^*}$}};         
\\
            \node[fill=macaroniandcheese](CC0) {1};  
            & \node[fill=non-photoblue](CC1) {0}; 
            & \node(CC2) {0}; 
            & \node(CC3) {1}; 
            %& \node[fill=macaroniandcheese](CC4) {1}; 
            & \node[draw=none] {\phantom{xxxxxxxxx}}; 
            & \node(CC5) {1};
            & \node(CC6) {1};    
            & \node(CC7) {1};   
            & \node(CC10) {1};             
\\        
            \node {0}; 
            & \node[fill=macaroniandcheese] {1}; 
            & \node {1}; 
            & \node {0}; 
            %& \node {1};
            & \node[draw=none] {\phantom{c}}; 
            & \node[fill=grannysmithapple] {0};
            & \node {1};            
            & \node {1}; 
            & \node {1}; 
\\        
            \node {0}; 
            & \node[fill=grannysmithapple]{0}; 
            & \node[fill=grannysmithapple] {0}; 
            & \node[fill=grannysmithapple] {0}; 
            %& \node[fill=grannysmithapple] {0}; 
            & \node[draw=none] {\phantom{c}}; 
            & \node[fill=grannysmithapple] {0};
            & \node[fill=grannysmithapple] {0};            
            & \node {1}; 
            & \node {1}; 
\\                
            \node[fill=grannysmithapple] {1};  
            & \node[fill=grannysmithapple] {0}; 
            & \node[fill=grannysmithapple] {0}; 
            & \node[fill=grannysmithapple] {0}; 
            %& \node[fill=grannysmithapple] {0}; 
            & \node[draw=none] {\phantom{c}}; 
            & \node[fill=grannysmithapple] {0};
            & \node[fill=grannysmithapple] {0};            
            & \node[fill=grannysmithapple] {0}; 
            & \node[fill=grannysmithapple] {0}; 
\\                 
            \node {1}; 
            & \node[fill=grannysmithapple]{0}; 
            & \node[fill=grannysmithapple] {0}; 
            & \node[fill=grannysmithapple] {0}; 
            %& \node[fill=grannysmithapple] {0}; 
            & \node[draw=none] {\phantom{c}}; 
            & \node {1};
            & \node[fill=grannysmithapple] {0};            
            & \node[fill=grannysmithapple] {0}; 
            & \node{1}; 
\\        
            \node {0};  
            & \node{1}; 
            & \node[fill=gray!20!white] {\phantom{1}}; 
            & \node[fill=gray!20!white] {\phantom{1}}; 
            %& \node {0};
            & \node[draw=none] {\phantom{c}}; 
            & \node {1};
            & \node {1};            
            & \node[fill=grannysmithapple] {0}; 
            & \node{1};             
\\        
            \node {1};  
            & \node{0}; 
            & \node[fill=gray!20!white] {\phantom{1}}; 
            & \node[fill=gray!20!white] {\phantom{1}}; 
            %& \node {1}; 
            & \node[draw=none] {\phantom{c}}; 
            & \node[fill=gray!20!white] {\phantom{0}};
            & \node {1};            
            & \node {1};  
            & \node {1}; 
\\ 
            \node[draw=none,fill=gray!20!white,minimum width=13.2pt,minimum height=20pt]  {\smash{\raisebox{-40\depth}{$\vdots$}}}; 
            & \node[draw=none,fill=gray!20!white,minimum width=13.2pt,minimum height=20pt]  {\smash{\raisebox{-40\depth}{$\vdots$}}};  
            & \node[draw=none,fill=gray!20!white,minimum width=13.2pt,minimum height=20pt]  {\smash{\raisebox{-40\depth}{$\vdots$}}}; 
            & \node[draw=none,fill=gray!20!white,minimum width=13.2pt,minimum height=20pt]  {\smash{\raisebox{-40\depth}{$\vdots$}}};  
            %& \node[draw=none,fill=gray!20!white,minimum width=13.2pt,minimum height=20pt]  {\smash{\raisebox{-40\depth}{$\vdots$}}};  
            & \node[draw=none] {\phantom{c}};
            & \node[draw=none,fill=gray!20!white,minimum width=13.2pt,minimum height=20pt]  {\smash{\raisebox{-40\depth}{$\vdots$}}};  
            & \node[draw=none,fill=gray!20!white,minimum width=13.2pt,minimum height=20pt]  {\smash{\raisebox{-40\depth}{$\vdots$}}};  
            & \node[draw=none,fill=gray!20!white,minimum width=13.2pt,minimum height=20pt]  {\smash{\raisebox{-40\depth}{$\vdots$}}};  
            & \node[draw=none,fill=gray!20!white,minimum width=13.2pt,minimum height=20pt]  {\smash{\raisebox{-40\depth}{$\vdots$}}}; 
\\  
            \node {1};  
            & \node[fill=gray!20!white]{\phantom{1}}; 
            & \node[fill=non-photoblue] {1}; 
            & \node[fill=non-photoblue] {1}; 
            %& \node {1}; 
            & \node[draw=none] {\phantom{c}}; 
            & \node {1};
            & \node {1};            
            & \node[fill=gray!20!white] {\phantom{1}};
            & \node {1}; 
\\             
            \node {0}; 
            & \node {1}; 
            & \node[fill=macaroniandcheese] {0}; 
            & \node[fill=macaroniandcheese] {1}; 
            %& \node {0}; 
            & \node[draw=none] {\phantom{c}}; 
            & \node {1};
            & \node {1};            
            & \node {1};  
            & \node{1}; 
\\            
            \node[fill=non-photoblue] {1};  
            & \node {1}; 
            & \node {1}; 
            & \node {1}; 
            %& \node[fill=non-photoblue] {0}; 
            & \node[draw=none] {\phantom{x}}; 
            & \node {1};
            & \node {1};            
            & \node {1}; 
            & \node{1}; 
\\ 
        };
    \draw[dotted, ->] (C0) -- (CC0); 
    \draw[dotted, ->] (C1) -- (CC1);
    \draw[dotted, ->] (C2) -- (CC2);
    \draw[dotted, ->] (C3) -- (CC3);
%    \draw[dotted, ->] (C4) -- (CC4);
    \draw[dotted, ->] (C5) -- (CC5);
    \draw[dotted, ->] (C6) -- (CC6);
    \draw[dotted, ->] (C7) -- (CC7);   
    \draw[dotted, ->] (C10) -- (CC10);     
    \node at (current bounding box.south)[anchor=south,xshift=-2.3cm,yshift=-0.7cm]{$(b)$};
    \node at (current bounding box.south)[anchor=south,xshift=2.3cm,yshift=0.0cm]{$(c)$};    
    \end{tikzpicture}
    \hspace{1.2cm}
    \begin{tikzpicture}[baseline=(current bounding box.north)]
        \matrix [nodes=draw,column sep=1mm, font=\small]
        {
        \node[draw=none,yshift=0.85cm]{};
        \\
            \node[draw=none,yshift=.5cm](C8) {\scriptsize{$\mathbf{[c_{0}]}$}};  
            &[-0.2cm] \node[draw=none,yshift=.5cm](C9) 
            {\scriptsize{$\mathbf{[c_{0}]} \wedge \mathbf{z^*}$}};
            \\
        
            \node(CC8) {1};
            & \node(CC9) {1};
            \\

            \node {0};
            & \node {0};
            \\
            
            \node {0};
            & \node {1};
            \\

            \node[fill=grannysmithapple] {1};
            & \node[fill=grannysmithapple] {0};
            \\   
            
            \node {1};
            & \node {1};
            \\

            \node {0};
            & \node {0};
            \\
            
            \node {1};
            & \node {1};
            \\

            \node[draw=none,fill=gray!20!white,minimum width=13.2pt,minimum height=20pt]  {\smash{\raisebox{-40\depth}{$\vdots$}}};  
            & \node[draw=none,fill=gray!20!white,minimum width=13.2pt,minimum height=20pt]  {\smash{\raisebox{-40\depth}{$\vdots$}}};  
            \\  

            \node {1};
            & \node {1};
            \\   
            
            \node {0};
            & \node {0};
            \\

            \node {1};
            & \node {1};
            \\              
        };
    \draw[dotted, ->] (C8) -- (CC8);
    \draw[dotted, ->] (C9) -- (CC9);         
%    \node at (current bounding box.south)[anchor=south,yshift=-1cm]{$\vdots$};        
    \node at (current bounding box.south)[anchor=south,xshift=-0.2cm, yshift=-0.7cm]{$(d)$};
\end{tikzpicture} 
\vspace*{-2ex}
\caption{\textit{\small
An illustration of the ultra-resilient properties described in Definition 
\ref{def:shift_k_unknown-fuzzy} for parameters $k = 4 \le n$ and $\alpha = 1$. 
In (a) an arbitrary subset 
$T = \{\mathbf{[c_0]}, \mathbf{[c_1]}, \mathbf{[c_2]}, \mathbf{[c_3]} \}$ 
of column vectors in the matrix is depicted. (b) shows a designated column 
$\mathbf{[c_j]} \in T$ (without loss of generality we assume that $\mathbf{[c_j]} = \mathbf{[c_0]}$), 
along with arbitrary shifts applied to the other columns in $T \setminus \{\mathbf{[c_0]} \}$.
%specifically 
%$\mathbf{[c_1]}(t-1), \mathbf{[c_2]}(2), \mathbf{[c_3]}(2), \mathbf{[c_4]}(0)$.
The initial and final bits of each original column vector are highlighted orange and cyan, 
respectively. 
%to aid in bit relocation after each shift. 
In (c) we have the superposition
%column 
vector 
$\mathbf{z} = \mathbf{[c_1]}(-1) \lor \mathbf{[c_2]}(2) \lor \mathbf{[c_3]}(2)$
surrounded by $\mathbf{z}(-1)$ on its left and $\mathbf{z}(1)$ on the right.
The rightmost column corresponds to the slipped vector
$\mathbf{z^*} = \mathbf{z}(-1) \lor \mathbf{z} \lor \mathbf{z}(1)$.
In (d) the vectors $\mathbf{[c_j]}$ and $\mathbf{[c_j]} \wedge \mathbf{z^*}$ 
are presented side by side. 
For $\alpha = 1$, the property ensures that 
$\left| (\mathbf{[c_{j}]} \wedge \mathbf{z^*} )_{[0, \tau(n,k)]} \right| <  \left|\mathbf{[c_j]}_{[0,\tau(n,k)]} \right|$
indicating the existence of at least one row where column vector $\mathbf{[c_{j}]}$ has a 1 while 
$\mathbf{[c_j]} \wedge \mathbf{z^*}$ has a 0, which in turn corresponds to 
$\mathbf{[c_j]}$ having a 1 while $\mathbf{z}(-1)$,  $\mathbf{z}$ and  $\mathbf{z}(1)$ 
all having a 0 in the same position (see the rows highlighted green in (b) and (c) and (d)).
}}
\label{fig:shiftSP}	
\end{figure}
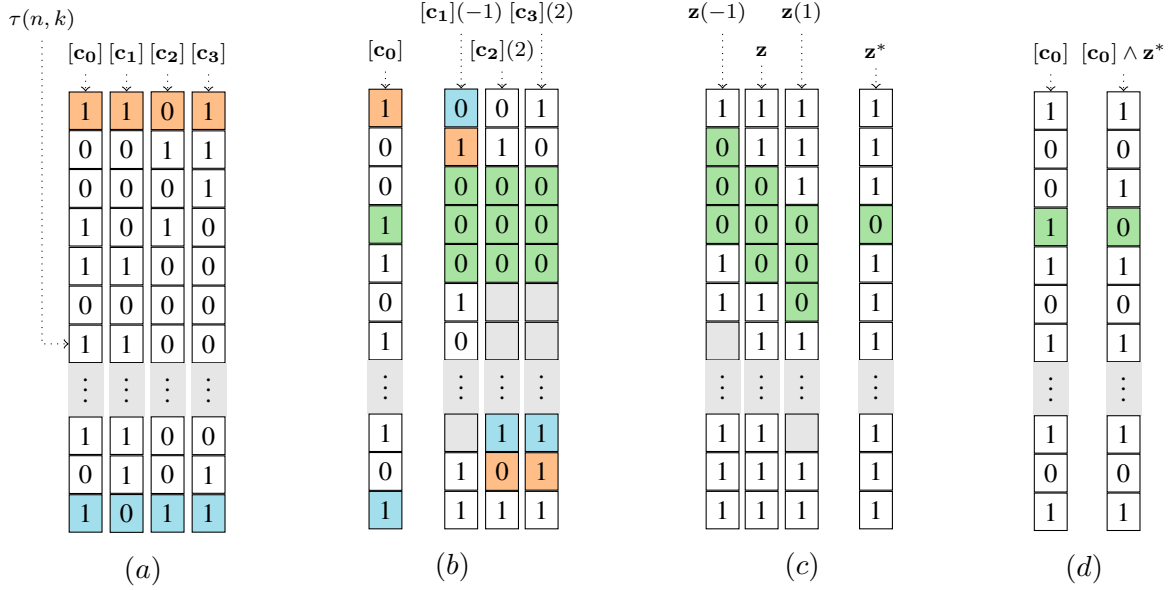

\section{Formal definition and notation}
\label{sec:preliminaries}

Given a binary vector $\x = (x_1, x_2, \ldots, x_t)$, we denote by $S(\x)$ 
the set of all \textit{cyclic shifts} of $\x$, that is, $S(\x)$ contains all 
different binary vectors of the form $(x_{1\oplus i}, x_{2\oplus i}, \ldots, x_{t\oplus i})$, 
where $\oplus$ denotes addition mod $t$ and $i=0, \ldots, t-1$. 
It is clear that $1 \leq |S(\x)| \leq t$.

For any binary vector $\x$, $|\x|$ represents the number of 1's in $\x$, 
also known as the \textit{weight} of $\x$. The symbols $\vee$ and $\wedge$ denote 
bitwise OR and AND operators, respectively, applied to binary vectors.

Let $R = \{\mathbf{y}_1, \mathbf{y}_2, \ldots, \mathbf{y}_r\}$ be a set of binary vectors. 
Given $R$, we 
%can 
construct 
the 
set $S_{\vee}(R)$~as~follows:

\vspace*{-0.5ex}
\begin{itemize}
    \item First we consider all cyclic shifts for each vector in $R$. 
\vspace*{-0.5ex}
    \item For each combination of cyclic shifts from the vectors in $R$, we perform a bitwise OR operation.
\end{itemize}

\vspace*{-0.5ex}
Formally, $S_{\vee}(R)$ is defined as:
\begin{center}
%\[
$   S_{\vee}(R) = 
   \left\{ 
         \bz = \bigvee_{i=1}^r \bz_i  \; \big|\;\;  \bz_i \in S(\by_i),  1 \le i \le r 
   \right\}$.
%.
%\]
\end{center}
In other words, $S_{\vee}(R)$ consists of all possible binary vectors obtained by 
taking the bitwise OR of one cyclic shift from each vector in $R$.

Given a $t \times n$ binary matrix $\mathbf{M}$, we refer to the $c_j$-th column vector of $\mathbf{M}$ as $\mathbf{[c_j]}$. For a column vector $\mathbf{[c_j]}$, 
we denote $\mathbf{[c_j]}(i)$, where $i=0, \ldots, t-1$, as the $i$th cyclic shift of $\mathbf{[c_j]}$. 
Formally, if $\mathbf{[c_j]} = (x_1, x_2, \ldots , x_t)$, then $\mathbf{[c_j]}(i) 
= (x_{1\oplus i}, x_{2\oplus i}, \ldots , x_{t\oplus i})$ for $i=0, \ldots, t-1$. 
For simplicity, we extend this definition to any integer $i \ge 0$, with the understanding that 
$i$ is taken modulo $t$ throughout the paper.
Finally, given any vector $\mathbf{x}$, its subvector 
from bit position $\beta_1$ to bit position $\beta_2$, where $0 \leq \beta_1 \leq \beta_2 \leq t$, is represented as $\mathbf{x}_{[\beta_1,\beta_2]}$.

%\subsection{ultra-resilient  superimposed codes for $k$ unknown}

We now present the formal definition of ultra-resilient superimposed codes. To do this, we first need to introduce two new concepts: code elongation, which generalizes the idea of code length for classic superimposed codes, and slipped vector, which is essential for preserving the isolation property.

\B
\paragraph{Code elongation.}
In our general scenario, where codewords can undergo arbitrary 
adversarial shifts, the unknown parameter $k$ introduces significant challenges that are not 
present in the classic case of fixed and aligned codewords. When codewords are fixed, 
a superimposed code for an unknown $k$ can be constructed by concatenating superimposed 
codes for known, incrementally increasing parameters $k' \le n$. However, it is well-known 
that when each codeword can experience an arbitrary adversarial shift, this concatenation technique, 
which attempts to `guess' the unknown parameter by concatenation,~becomes~ineffective.
As a result, the concept of \textit{code length}, which for a known $k$ (and a given $n$) was a 
fixed value corresponding to the number $t$ of rows in the matrix, evolves into the more general 
notion of \textit{code elongation} in our much broader definition of ultra-resilient superimposed 
codes for unknown $k$.
Code \elongation\ is characterized by a function $\tau: \mathbb{N} \times \mathbb{N} \to [0,t)$ 
that, in addition to the given $n$, also depends on the unknown parameter $k$. 
This dependence enables the code to maintain its ultra-resilience properties across different 
(unknown) column subset sizes $k$.

\B
\paragraph{Slipped vector.}
For any vector $\mathbf{z}$, we define the corresponding \textit{slipped} vector as 
$\mathbf{z^*} = \mathbf{z}(-1) \lor \mathbf{z} \lor \mathbf{z}(1)$.
The slipped vector is crucial for ensuring the isolation property in the
definition of ultra-resilient superimposed codes given below 
(see Figure \ref{fig:shiftSP}
%\dk{in the beginning of the Appendix} 
for a graphical reference): 
any 1-bit in $\mathbf{[c_{j}]}$ that does not overlap with 
$\mathbf{z^*}$, i.e., in the surplus 
$\mathbf{[c_{j}]} \setminus (\mathbf{[c_{j}]} \wedge \mathbf{z^*})$, 
indicates the existence of a row $i$ such that 
$\mathbf{[c_{j}]}$ has a 1 at position $i$ while $\mathbf{z}$ has 0's in positions 
$(i-1) \bmod t$, $i$, and $(i+1) \bmod t$.

\begin{definition}[Ultra-resilient superimposed code]
\label{def:shift_k_unknown-fuzzy}
Let $n$ be any integer.
Given a function $\tau: \mathbb{N} \times \mathbb{N} \to [0,t)$, where $t \ge n$, 
and a real number $0 < \alpha \le 1$, we say that  
a $t \times n$ binary matrix $\mathbf{M}$ is a $(n, \alpha)$-ultra-resilient   
superimposed code (denoted $(n, \alpha)$-URSC) of \elongation\ $\tau$, if the following condition~holds: 

For any $2\le k\le n$ and any subset $T$ of column indices of $\mathbf{M}$ with $|T| = k$, 
and for any column index $c_j \in T$, the inequality 
$$
\left| \Big(\mathbf{[c_{j}]} \wedge \mathbf{z^*} \Big)_{[0, \tau(n,k)]} \right| < \alpha \cdot \left|\mathbf{[c_j]}_{[0,\tau(n,k)]} \right|
$$
is satisfied for all $\mathbf{z} \in S_{\vee}(T \setminus \{c_j\})$. 

\end{definition}

%Observe that each ultra-resilient superimposed code (\textit{i.e.}, code satisfying 
%Definition~\ref{def:shift_k_unknown-fuzzy}) satisfies Definition~\ref{def:shift_k_unknown} 
%of resilient superimposed codes, because
%\[
%\Big(\mathbf{c} \wedge \big(\mathbf{z}(t-1) \lor \mathbf{z} \lor \mathbf{z}(1)\big)\Big)_{[0, \tau(n,k)]} 
%=
%\big(\mathbf{c} \wedge \mathbf{z}(t-1)\big)_{[0, \tau(n,k)]} \lor \big(\mathbf{c} \wedge \mathbf{z}\big)_{[0, %\tau(n,k)]} \lor \big(\mathbf{c} \wedge \mathbf{z}(1)\big)_{[0, \tau(n,k)]}
%\ ,
%\]
%is not true.
%may not be true.

%\gdm{The following sentence needs a revision.}
%Note also that the analogous notion 
%of ultra-resilient invariant superimposed codes could be 
%implemented for known $k$, based on the Definition~\ref{def:relaxed}, but we omit it here as 
%we could use codes satisfying 
%the more general Definition~\ref{def:shift_k_unknown-fuzzy} for the purpose of the considered applications.

%\subsection{Basic properties of shift-invariant superimposed codes}

\section{Construction of Ultra-Resilient Superimposed Codes (URSC)}
\label{sec:unknown-k}

This section focuses on the construction of ultra-resilient superimposed codes without knowing  
the parameter $k$. The objective is to design a randomized algorithm that, 
given the input parameters $n$ and $\alpha$ and for any $\epsilon > 0$,
efficiently generates an ultra-resilient superimposed code with \elongation\ 
$\tau(n,k) = c (k^2/ \alpha^{2+\epsilon})\ln n$, for any (unknown) $1 < k \le n$.
This is near-optimal in view of the $\Omega((k/\alpha)^2(\log_{k/\alpha} n)$ lower bound
proved in \cite{RV2024}.
This lower bound also implies that to ensure the code elongation remains within practical limits, 
specifically polynomial in $k$, one can reasonably assume that 
$e^{-k} < \alpha \le 1$.

Although the construction algorithm is randomized, the generated code can be used reliably in 
deterministic algorithms. Additionally, the code only needs to be generated once, as it can be reused 
in different contexts as long as the parameters $n$ and $\alpha$ remain unchanged.

Our approach to efficiently constructing the codes revolves 
around two key concepts: the \textit{Collision Bound Property} and the strategic selection of 
\textit{assignment probabilities} for the 1's and 0's in the matrix. These concepts are closely 
intertwined: the Collision Bound Property streamlines the computational verification of code correctness, 
while the strategic assignment probabilities guarantee that the matrix satisfies the Collision Bound Property.

In Section \ref{CBpropertysec}, we introduce the Collision Bound Property and we demonstrate its sufficiency 
in ensuring that a given matrix qualifies as an ultra-resilient superimposed code.
Section \ref{sub:random} introduces our random matrix construction method. Subsequently, 
Section \ref{satisfaction} serves as the main technical segment where 
we establish that matrices generated using our random procedure have a high probability of satisfying 
both inequalities stipulated by the Collision Bound Property and consequently of qualifying as
ultra-resilient superimposed codes.
Finally, in Section \ref{construction}, we outline the construction algorithm. This algorithm efficiently 
utilizes repeated applications of our random procedure of Section \ref{sub:random} 
to generate ultra-resilient superimposed codes.

Throughout this section, we assume that the two parameters $n$ and $\alpha$ are fixed and given. 
Specifically, $n$ is an integer such that $n \geq 2$, and $\alpha$ is a real number such 
that $e^{-k} < \alpha \leq 1$.

\subsection{Collision Bound Property}\label{CBpropertysec}

The definition of 
%ultra-resilient superimposed 
\dk{URSC}
codes involves a condition on subsets of $k$ columns, 
which can be computationally challenging to verify due to the super-polynomial number of such subsets. 
A crucial step towards an efficient construction is the introduction 
of the \textit{Collision Bound Property}, a sufficient condition for ensuring the resilience properties 
of $(n, \alpha)$-URSC, which applies to pairs of columns rather than subsets of $k$ columns.

Before delving into the formal definition, let us summarize the Collision Bound Property. 
This property divides each column into upper and lower segments 
-- specifically $[0, \tau_1(n,k)]$ and $[\tau_1(n,k), \tau_2(n,k)]$, with $\tau_2(n,k)$
corresponding to the elongation of the code -- and ensures two key inequalities: 
the \textit{Weight Inequality}, which pertains to any individual column, 
and the \textit{Collision Weight Inequality}, which applies to any pair of columns.
\vspace*{-1ex}
\begin{itemize}
    \item \textbf{Weight Inequality}: The first inequality compares the weights of the upper and lower segments 
    of any column. Specifically, it ensures that the weight of the upper segment is at most $\alpha$ times 
    the weight of the lower segment. 
    This establishes a specific dominance of the lower segment's weight over the 
    upper segment's weight.
    
\vspace*{-0.5ex}
    \item \textbf{Collision Weight Inequality}: The second inequality bounds the ``collision weight" 
    of any pair~of~columns:
    \[
    \left|\Big(\mathbf{[c_j]} \land \big(\mathbf{[c_{j'}]}(i-1) 
                           \lor \mathbf{[c_{j'}]}(i) 
                           \lor \mathbf{[c_{j'}]}(i+1) \big)
                           \Big)_{[\tau_1(n, k), \tau_2(n, k)]}
    \right| \ ,
    \]
    defined as the number of positions in the lower segment where a column 
    intersects with the slipped vector of any cyclic shift of the other column of the pair. 
    This inequality ensures that this collision 
    weight does not exceed $\alpha$ times one $(k - 1)$th 
    of the lower segment's weight. This helps in 
    controlling the overlap between columns, which is crucial for maintaining the superimposed code~property.
\end{itemize}

%We can now formally define the property.

\begin{definition}[Collision Bound Property]\label{CBproperty}
%Let $\alpha$ be a real number, $0 < \alpha \leq 1$, and 
Let $\mathbf{M}$ be a $t \times n$ 
binary matrix for some integer $t \ge n$. Let $\tau_1, \tau_2: \mathbb{N} \times \mathbb{N} \to [0,t)$ 
be two integer functions. 

\medskip
\textbf{Collision Bound Property} \pairCondition: For every $1 < k \leq n$, for each pair of column indices $c_j$ and $c_{j'}$ and every cyclic shift $\mathbf{[c_{j'}]}(i)$, $0 \leq i \leq t-1$, both of the following inequalities hold:

\begin{itemize}%[leftmargin=5mm]
    \item \textbf{Weight Inequality}:
    \begin{equation}\label{conjineq}
        |\mathbf{[c_j]}_{[0, \tau_1(n, k) ]} |
        \le \alpha\,|\mathbf{[c_j]}_{[\tau_1(n, k), \tau_2(n, k)]}| \ ,
    \end{equation}
    \item \textbf{Collision Weight Inequality}:
    \begin{equation}
    \label{conjineq2}
    \hspace*{-2em}
        \left|\Big(\mathbf{[c_j]} \land \big(\mathbf{[c_{j'}]}(i-1) 
                               \lor \mathbf{[c_{j'}]}(i) 
                               \lor \mathbf{[c_{j'}]}(i+1) \big)
                               \Big)_{[\tau_1(n, k), \tau_2(n, k)]}
        \right| 
        \leq \left\lfloor \frac{\alpha  \; |\mathbf{[c_j]}_{[\tau_1(n, k), \tau_2(n, k)]}|  - 1}{k - 1} \right\rfloor.
    \end{equation}     
\end{itemize}  
\end{definition}

The following lemma establishes the sufficiency of the Collision Bound Property in guaranteeing that
a matrix $\mathbf{M}$ is an $(n, 2\alpha)$-URSC.

%Maybe we can assume $\mathcal{P} (\mathbf{M}, \alpha/2, \tau_1, \tau_2)$ in the hypothesis
%instead of $\mathcal{P} (\mathbf{M}, \alpha, \tau_1, \tau_2)$.
\begin{lemma}\label{1e3}
Let $\mathbf{M}$ be a $t \times n$ binary 
matrix for some integers $t \ge n$. Assume that $\tau_1, \tau_2: \mathbb{N} \times \mathbb{N} \to [0,t)$  
are two integer functions such that the Collision Bound Property \pairCondition\ is satisfied.
Then, the matrix $\mathbf{M}$ 
{$(n, 2\alpha)$}-URSC with \elongation\ $\tau_2(n,k)$.
\end{lemma}

\begin{proof}
	Let us consider two functions $\tau_1$, $\tau_2$ and a $t\times n$ binary matrix $\mathbf{M}$ satisfying property \pairCondition.
    For any $1 < k \le n$,
	%for some pair of integers $\tau_1$ and $\tau_2$, 
	%with $0 \le \tau_1 < \tau_2 \le t$.
	consider a fixed subset $T$ of $|T| = k$ column indices, and a column index $c_j \in T$.  
	Let $\mathbf{z} \in S_{\vee}(T \setminus \{c_j\})$
	and $c_{m_1}, \ldots, c_{m_{k-1}}$ be the column indices in 
	$T \setminus \{c_j\}$. 
	By definition, 
	 $\mathbf{z} = \mathbf{[c_{m_1}]}(i_1) \lor \cdots \lor \mathbf{[c_{m_{k-1}}]}(i_{k-1})$,
	 for some arbitrary shifts $0 \le i_1, \ldots, i_{k-1} \le t-1$.
	 Hence, applying the distribution of conjunction over disjunction, we get
\begin{equation*}
(\mathbf{[c_j]} \land \mathbf{z})_{[\tau_1(n,k),\tau_2(n,k)]}  
	=  		 
 \left(\mathbf{[c_j]} \land \mathbf{[c_{m_1}]}(i_1) \right)_{[\tau_1(n,k),\tau_2(n,k)]} 
	\lor  \cdots \lor
	\left(\mathbf{[c_j]} \land \mathbf{[c_{m_{k-1}}]}(i_{k-1}) \right)_{[\tau_1(n,k),\tau_2(n,k)]}\ .
\end{equation*} 
%where $\tau_1=\tau_1(n,k)$ and $\tau_2=\tau_2(n,k)$ are the two row indices guaranteed by the hypothesis
%of the lemma. 
By inequality (\ref{conjineq2}) of property \pairCondition, it follows that
\begin{equation}\label{eq:weight}
	|(\mathbf{[c_j]} \land \mathbf{z})_{[\tau_1(n,k),\tau_2(n,k)]} | 
	   \le  (k-1)   \left\lfloor \frac{\alpha \; |\mathbf{[c_j]}_{[\tau_1(n,k),\tau_2(n,k)]}|  - 1}{k - 1} \right\rfloor
	   < \alpha\, | \mathbf{[c_j]}_{[\tau_1(n,k),\tau_2(n,k)]} |
    \ .
\end{equation}
%Fix any $\tau \ge \tau_2$. 
Considering also the second inequality of \pairCondition, we have:
\begin{eqnarray*}
	\left| (\mathbf{[c_{j}]} \wedge \mathbf{z})_{[\tau_2(n,k)]} \right| 
		&\le &  | \mathbf{[c_j]}_{[0, \tau_1(n,k) ]}| 
		           + \left| (\mathbf{[c_j]} \land \mathbf{z} )_{[\tau_1(n,k),\tau_2(n,k)]} \right| \\
		           %+  | \mathbf{[c_j]}_{[\tau_2(n,k) + 1, \tau_2(n,k)]}| \\
		&< &       \alpha\,| \mathbf{[c_j]}_{[\tau_1(n,k),\tau_2(n,k)]} |
		           + \alpha\,| \mathbf{[c_j]}_{[\tau_1(n,k),\tau_2(n,k)]} | \;\;\;\;\; 
             \text{ by (\ref{conjineq}) and (\ref{eq:weight})}\\
		          % +  | \mathbf{[c_j]}_{[\tau_2(n,k) + 1, \tau_2(n,k)]}|  \;\;\text{ by (\ref{eq:weight})}\\
		&< & { 2\alpha  | \mathbf{[c_j]}_{[\tau_2(n,k)]}| } 
  \ .
\end{eqnarray*}

Since we have established this property for all $\mathbf{z} \in S_{\vee}(T \setminus c_j)$, and 
we have shown it holds for all possible subsets 
 $T$ and column indices $c_j$ satisfying $|T| = k$, we conclude that $\mathbf{M}$ is a 
{$(n,2\alpha)$}-URSC with \elongation\ $\tau_2(n,k)$.
\end{proof}

\subsection{Random construction}\label{sub:random}

In this subsection, we present a random construction of a binary matrix. 
As demonstrated in the next subsection, this construction has a high probability of 
satisfying the 
Collision Bound Property \pairCondition\
and, hence, of being an $(n, 2\alpha)$-URSC
in view of Lemma \ref{1e3}. In Theorem \ref{thm:code-final} we finally obtain 
an $(n, \alpha)$-URSC.

In addition to $n$ and $\alpha$, the construction relies on two more parameters: 
an arbitrarily small constant $\epsilon > 0$ and a real constant $c > 0$. 
The role of $\epsilon$ is to bring the elongation of the code arbitrarily close to the asymptotic bound 
$O((k/\alpha)^2 \ln n)$. The constant $c$, if sufficiently large, ensures a high probability 
of successful construction, as will be demonstrated in the next section.

The goal of the random construction that we are going to describe is to assign the probability $p(r)$ 
of having a 1 in the $r$th bit of each column, for $0 \le r < t$. The ignorance of the parameter $k$ 
presents new significant challenges (in addition to those related to the arbitrary shifts), 
as we cannot use $k$ in assigning the probabilities $p(r)$, nor can we use a uniform distribution. 
Instead, we estimate $k$ as we descend the positions of the columns, with probabilities in the upper part 
tailored for smaller values of $k$ and those in the lower part for larger values of $k$. 
This is achieved by gradually decreasing the probabilities $p(r)$ as $r$ increases, \textit{i.e.}, 
as we move down the positions of the columns. Throughout this process, we must ensure that the two 
inequalities of the Collision Weight Property are satisfied, creating a subtle trade-off as explained below.

The Weight Inequality requires that the decrease in the frequency of 1's is controlled such that the weight 
of the lower segment (from $\tau_1(n,k)$ to $\tau_2(n,k)$) dominates the weight of the upper segment 
(up to position $\tau_1(n,k)-1$). 
Conversely, the Collision Weight Inequality requires that this dominance does not lead to an excessive number of collisions in the lower segment; specifically, the number of collisions in the lower (dominating) segment must be significantly less than the weight of each column. 
Additionally, this inequality must be satisfied without extending $\tau_2(n,k)$ beyond the desired elongation.

%(which, as we will see, decrease with the square of the probabilities) 

%As we will show, the carefully chosen probability function 
%$p(r) = \sqrt{\frac{1}{\lfloor r/\ln n \rfloor +1}}$ strikes the right balance among all these conflicting targets.
As we will show, a carefully chosen probability function
that decreases according to the square root of the inverse of the column's position $r$,
strikes the right balance among all these conflicting targets.

Finally, it is worth noting, that although each column of the matrix will have length 
$t = c/(\alpha^{2+\epsilon}) n^2 \ln n$ 
(since they cannot depend on $k$), 
it will be shown 
%in the next section 
later
that 
the code still guarantees an elongation of 
$\tau(n,k) = c/(\alpha^{2 + \epsilon}) k^2 \ln n$.

\begin{definition}[Random Matrix Construction ${\mathcal{M}}(n,\alpha, \epsilon, c)$]\label{randomatrix} 
Let us define a random matrix ${\mathcal{M}}(n, \alpha, \epsilon, c)$ of $n$ columns and 
$t = (c/\alpha^{2 + \epsilon}) n^2\ln n$ rows, generated using the following procedure. 
	The $r$th bit of each column, $0\le r < t$, is independently set to $1$ with a probability given by:
	\[
	p(r) = \sqrt{\frac{1}{\lfloor r/\ln n \rfloor +1}}
 \ , 
	\]
	and to $0$ with the complementary probability.
This corresponds to each column being partitioned into $(c/\alpha^{2+\epsilon}) n^2$ blocks 
	of equal length $\ln n$, with every bit of the $b$th block, 
 for $0 \le b < (c/\alpha^{2 + \epsilon}) n^2$, independently set to $1$
 with a probability given by $1/\sqrt{b + 1}$,
	and to $0$ with the complementary probability.	
\end{definition}

\subsection{Satisfying the Collision Bound Property}\label{satisfaction}

Our next objective is to show that for $\tau_1(n,k) = \frac{c}{64 } k^2 \ln n$ and 
$\tau_2(n,k) = \frac{c}{\alpha^{2 + \epsilon}} k^2 \ln n$, any random matrix constructed as
illustrated in subsection \ref{sub:random}, satisfies the Collision Bound Property
\pairCondition\
with high probability.
(It is important to clarify that, as we will see in Section \ref{construction}, the functions 
$\tau_1$ and $\tau_2$, which are defined in terms of the unknown $k$ (and $n$), do not need 
to be known by the construction algorithm.)
Namely, we will be proving the following.
\begin{lemma}
\label{lep} 
Fix any $\epsilon > 0$. Define $\tau_1(n,k) = (c /64 ) k^2 \ln n$ and 
$\tau_2(n,k) = (c/\alpha^{2+\epsilon} ) k^2 \ln n$, where 
$c > 0$ is a sufficiently large real constant. Let 
$\mathbf{M} = \mathcal{M}(n, \alpha, c)$ be a random matrix. 
For any given $1 < k \leq n$, a pair of column indices $c_j$ and $c_{j'}$, 
and a cyclic shift $\mathbf{[c_{j'}]}(i)$ with $0 \leq i \leq t-1$, 
 the probability that the Collision Bound Property \pairCondition\ does not hold is less than
%  \[
$  \frac{6}{c^2}\cdot  n^{-8\ln\left(\frac{4}{\alpha}\right)}$. 
%\ .
% \]
\end{lemma}

The detailed technical proof of Lemma \ref{lep} is deferred to Appendix \ref{proofs}. The proof of Lemma \ref{lep} is quite involved, as it requires bounding the probabilities associated with satisfying the Weight Inequality and the Collision Weight Inequality separately (subsections \ref{WIsub} and \ref{CWIsub} of Appendix \ref{proofs}, respectively).

The Weight Inequality requires us to control the expected weights of both the upper and lower segments of each column $\mathbf{[c_j]}$. We analyze these segments as expected values over the randomized matrix construction, allowing us to bound the probability that the Weight Inequality is satisfied (see Lemma \ref{lep1}).

For the Collision Weight Inequality, we define a random variable to represent the number of collisions between each pair of columns $\mathbf{[c_j]}$ and $\mathbf{[c_{j'}]}$ within the lower segment: $[\tau_1(n,k), \tau_2(n,k)]$. This random variable must account for any possible cyclic shift $0 \le i < t$ 
(and corresponding slipped vector) of the second column $\mathbf{[c_{j'}]}$ relative to the first column $\mathbf{[c_j]}$. Specifically, for any fixed shift $0 \le i < t$, this random variable is given by 
$\sum_{d=-1}^1 |\mathbf{[c_j]} \land \mathbf{[c_{j'}]}(i + d)|$. Bounding the probability that this 
variable meets the inequality’s criteria is addressed in Lemma \ref{CWIup} and is particularly challenging
for the following reasons.

The expected value of this variable is heavily influenced by the shift magnitude of the second column of the pair, which significantly affects, within the lower interval $[\tau_1(n,k), \tau_2(n,k)]$, the probabilities 
$p\left((r+i+d) \bmod t\right)$, representing the probability that the slipped vector of the
shifted column hosts a 1 in its $r$th position after a shift $i$, for $0 \le i < t$ and $d = -1, 0, 1$. 
To address this, Lemma \ref{CWIup} 
analyzes three separate cases based on the magnitude of probabilities $p\left((r+i+d) \bmod t\right)$ for all $r$ in the interval $\tau_1(n,k) \leq r \leq \tau_2(n,k)$.

The first case handles collisions when the shift of the second column in the pair is minimal, 
causing all probabilities of the shifted column to be high within the positions of the lower segment.
Here, we apply the \textit{rearrangement inequality} \cite{HLP1934} (see the appendix), which enables 
a precise estimation of the expected number of collisions for high-probability entries. 
This technique maximizes the sum of products of corresponding probabilities 
(i.e., the expected collision count) by aligning the largest probability values in both columns, 
yielding a tight upper bound on collisions under minimal shifts.

In the second case, we consider mid-range shifts where the probabilities in the shifted column remain 
within a constant factor of each other. 
To handle this, we develop a \textit{pairwise bounding technique} that leverages the near-uniformity 
of probabilities across positions. This approach minimizes the dependency on exact probabilities, 
allowing us to obtain a balanced estimate of collisions despite moderate variations. 

Finally, the third case considers the situations where the probabilities of the 
shifted column are very low from $\tau_1(n,k)$ up to some intermediate position $\tau'$, with
$\tau_1(n,k) < \tau' < \tau_2(n,k)$, and then become large from the next position $\tau' +1$ until 
$\tau_2(n,k)$. 
To address this, we employ a \textit{two-segment analysis} that treats the low-probability 
and high-probability segments separately. Specifically, we combine the techniques from Case 2 
for the low-probability segment and Case 1 for the high-probability segment.

\iffalse
The second case addresses mid-range shifts, where the probabilities in the shifted column remain within a 
constant factor of each other. Finally, the third case considers situations where the probabilities of the 
shifted column are very low from $\tau_1(n,k)$ up to some intermediate position $\tau'$, with
$\tau_1(n,k) < \tau' < \tau_2(n,k)$, and then become large from the next position $\tau' +1$ until 
$\tau_2(n,k)$.
\dk{???WHERE WE USE REARANGEMENT INEQUALITY???}
\fi

The proof of Lemma \ref{lep} is ultimately completed by combining the probability of satisfying the Weight Inequality (Lemma \ref{lep1}) and that of satisfying the Collision Weight Inequality (Lemma \ref{CWIup}).

%{CWIup}
%The value of $E[X^{d}_{j,j'}(i)]$ is significantly influenced by the size of $p\left((r+i+d) \mod t\right)$, 
%which can vary considerably within the interval $[\tau_1(n,k), \tau_2(n,k)]$ depending on the shift 
%$0 \le i < t$. 

%\bigskip

\begin{algorithm}[t!]
\small
\caption{\small Compute ultra-resilient superimposed codes ($k$ is unknown)}\label{alg:k_unknown}
\linespread{1.1}\selectfont 
\textbf{Input:} an integer $n \geq 2$, a real number $0 < \alpha \leq 1$, an arbitrarily small constant 
$\epsilon > 0$ and a 
%real 
constant $c >0$.\\
\textbf{Output:} a matrix $\mathbf{M}$ that is a $(n, 2\alpha)$-URSC with \elongation\ $\tau(n,k) = c (k^2/\alpha^{2+\epsilon})\ln n$, for any $1 < k \le n$.

\SetKwFunction{MyFunction}{Check\_$\mathcal{P}$}
\SetAlgoNlRelativeSize{0}
\SetAlgoNlRelativeSize{-1}

\vspace*{0.3ex}
%\BlankLine
    Let $\tau_1$ and $\tau_2$ be functions defined as: 
    $\tau_1(x,y) = (c/64)x^2 \ln y$, 
    $\tau_2(x,y) = c (x^2/\alpha^{2+\epsilon})\ln y$. \label{functions}

%\BlankLine
\vspace*{0.3ex}
\Repeat{\MyFunction{$\mathbf{M}$, $\alpha$, $\tau_1$, $\tau_2$} $= $ {\sc true}}{
    Generate a random matrix $\mathbf{M} = {\mathcal{M}}(n,\alpha, \epsilon, c)$
     (see Definition \ref{randomatrix})\; \label{random_gen}
}

\KwRet {the generated matrix $\mathbf{M}$}\;

%\BlankLine
\BlankLine

\SetKwProg{Fn}{Function}{:}{}
\Fn{\MyFunction{$\mathbf{M}$, $\alpha$, $\tau_1$, $\tau_2$}}{
    %\tcp{Function definition}
    \For{$1 < k \leq n$}{                   \label{for_start}
        \For{each pair of column indices $c_j$ and $c_{j'}$ of $\mathbf{M}$}{
            \For{every shift $\mathbf{[c_{j'}]}(i)$}{
                \If {\label{ifcond}
                either the Weight Inequality (\ref{conjineq}) or the Collision Weight Inequality, 
                or both, are not satisfied }
                {
                     \KwRet {\sc false;} \tcp{ if the Collision Bound Property is not satisfied, 
                     the function terminates and returns {\sc false} }                        
                }
            }
        }
    }                                     \label{for_end}
    \KwRet {\sc true}\; 
}
\end{algorithm}

\subsection{The construction algorithm}\label{construction}

The construction of the ultra-resilient superimposed codes is accomplished by Algorithm \ref{alg:k_unknown}, 
a randomized algorithm that, in addition to the parameter $n$ and the real number $0 < \alpha \le 1$,
takes as input an arbitrarily small constant $\epsilon > 0$ and a constant $c > 0$. 
The next lemma proves that if the constant $c$ is  
sufficiently large as established in Lemma \ref{lep}, then Algorithm \ref{alg:k_unknown}
outputs an $(n, 2\alpha)$-URSC with \elongation\ 
$\tau(n,k) = c (k^2/\alpha^{2+\epsilon})\ln n$, for any $1 < k \le n$.

\begin{lemma}
\label{thm:unknown-k}
For 
%an integer $n \geq 2$ and 
$\alpha\in(0,1]$, Algorithm \ref{alg:k_unknown} generates with high probability in polynomial time 
an $(n, 2\alpha)$-URSC with \elongation\ 
$\tau(n,k) \le \frac{c\cdot k^2}{\alpha^{2+\epsilon}} \ln n$, for any $1 < k \le n$. 
The same 
result 
can 
%also 
be obtained~in~expectation.
\end{lemma}

\begin{proof}
The algorithm sets functions $\tau_1(x,y) = (c/64 )x^2 \ln y$
and $\tau_2(x,y) = c (x^2/\alpha^{2+\epsilon})\ln y$.
Then it uses a \texttt{repeat-until} loop, which at each iteration generates a random matrix 
$\mathbf{M} = {\mathcal{M}}(n,\alpha, \epsilon, c)$, according to the random construction of  
Definition \ref{randomatrix}, and invokes a function 
\texttt{Check}\_$\mathcal{P}(\mathbf{M}, \alpha, \tau_1, \tau_2)$.
It continues iterating until this function returns {\sc true}.

The role of \texttt{Check}\_$\mathcal{P}(\mathbf{M}, \alpha, \tau_1, \tau_2)$ is simply to verify 
the satisfaction of the Collision Bound Property $\mathcal{P}(\mathbf{M}, \alpha, \tau_1, \tau_2)$. 
This verification is ensured by the three nested \texttt{for} loops in lines \ref{for_start} - \ref{for_end}. 
Namely, the function returns {\sc true} if and only if, 
for every $1 < k \leq n$, for every pair of column 
indices $c_j$ and $c_{j'}$, and every cyclic shift $\mathbf{[c_{j'}]}(i)$ with $0 \leq i \leq t-1$, 
both the Weight Inequality (\ref{conjineq}) and the Collision Weight Inequality (\ref{conjineq2})
hold.
This indeed corresponds to the Collision Bound Property 
$\mathcal{P}(\mathbf{M}, \alpha, \tau_1, \tau_2)$ being satisfied. 
Consequently, from Lemma \ref{1e3} it follows that when the execution exits the \texttt{repeat-until} loop,
the last generated matrix $\mathbf{M}$ is a valid 
$(n,\alpha)$-URSC with \elongation\ 
$\tau(n,k) \le (c/\alpha^{2+\epsilon}) k^2 \ln n$, 
for any $1\le k \le n$.

To complete the proof, we must demonstrate that the total number of operations involved is polynomial 
with high probability and in expectation. Since the number of operations required to construct the 
random matrix is evidently polynomial — specifically, 
$n \cdot t = n \cdot \left(\frac{c}{\alpha^{2+\epsilon}} n^2 \ln n \right)$ random choices, 
one for each bit of the matrix — it suffices to show the following. 
First, we prove that each iteration requires polynomial time. Then, we establish that the number 
of iterations is constant both with high probability and in expectation.

\medskip
\textbf{Number of operations per iteration.}
In each iteration of the 
\texttt{repeat-until} loop, 
the function 
\texttt{Check}\_$\mathcal{P}(\mathbf{M}, \alpha, \tau_1, \tau_2)$ is invoked to verify whether the matrix satisfies 
property $\mathcal{P}(\mathbf{M}, \alpha, \tau_1, \tau_2)$. 
In this function, the three nested loops cause the \texttt{if} condition on 
line \ref{ifcond} to be checked at most 
$n\cdot 2\binom{n}{2} \cdot t$ 
times. 
Checking the \texttt{if} condition requires verifying inequalities (\ref{conjineq}) and (\ref{conjineq2}), 
which takes fewer than $3t$ operations. 

Specifically, the first inequality involves counting the number of 1's in a vector $\mathbf{[c_j]}$ 
over no more than $t$ positions and then comparing two values. 
The second inequality requires counting the number of 1's in a 
vector $\left( \mathbf{[c_j]} \land (\mathbf{[c_{j'}]}(i-1) \lor \mathbf{[c_{j'}]}(i) \lor \mathbf{[c_{j'}]}(i+1) ) \right)$ and in a vector $\mathbf{[c_j]}$, both over at most $t$ positions, 
followed by comparing two values.

Overall, any iteration of the \texttt{repeat-until} loop takes time no more than 
$n\cdot 2\binom{n}{2} \cdot t \cdot 3 t < 
3 n^3 t^2$. 

\medskip
\textbf{Number of iterations.}
The remaining task is to count the total number of iterations of the algorithm until termination.
Let's consider the probability that the function 
\texttt{Check}\_$\mathcal{P}(\mathbf{M}, \alpha, \tau_1, \tau_2)$, invoked on a random matrix $\mathbf{M} = \mathcal{M}(n,c)$ 
with $\tau_1$ and $\tau_2$ set as specified in line \ref{functions}$,$ returns {\sc false} during a specific iteration of the innermost loop. 
This corresponds to the probability that, for any fixed value of $k$ in the first loop, 
any pair of columns fixed in the second loop, and any shift $\mathbf{[c_{j'}]}(i)$ fixed in the third loop, 
the inequalities (\ref{conjineq}) and (\ref{conjineq2}) do not hold.              
Assuming that the input constant $c$ used to build the matrix $\mathbf{M} = {\mathcal{M}}(n,c)$ 
is sufficiently large, 
according to Lemma \ref{lep}, the probability that in any of the iterations of the three nested loops the inequalities do not hold, is at most 
$\frac{6}{c^2}\cdot  n^{-8\ln\left(\frac{4}{\alpha}\right)}$. 
Therefore, by applying the union bound over all three nested loops, we can calculate 
the probability that the function returns {\sc false} in any of its iterations as follows:
\[
%\begin{eqnarray*}
    3 n^3 t^2 
        \left(\frac{6}{c^2}\cdot  n^{-8\ln\left(\frac{4}{\alpha}\right)}\right)
    %& = & 
    \ = \ 
    \frac{3 c^2 n^7 \ln^2 n}{\alpha^{4+2\epsilon}} 
         \left(\frac{6}{c^2}\cdot  n^{-8\ln\left(\frac{4}{\alpha}\right)}\right) 
         %\\
    %& = & 
    \ = \ 
    \frac{18 \ln^2 n}{\alpha^6} \cdot n^{7 - 8\ln(4/\alpha)}  
    %\\   
    %& < & 
    \ < \ 
    \frac{1}{n^2}
    \ ,
%\end{eqnarray*}
\]
where the last inequality holds for $n \ge 9$.
Therefore, for sufficiently large $n$, a single iteration of the \texttt{repeat-until} loop 
is sufficient to obtain a valid $(n,\alpha)$-URSC 
with high probability. Given that we have demonstrated that each iteration requires polynomial time, 
it follows that the algorithm concludes within polynomial time with high probability.

Analogously, to demonstrate that the result also holds in expectation, it suffices to show that the expected number of iterations of the \texttt{repeat-until} loop is constant.
The probability that the algorithm terminates at the $i$-th iteration is the probability that, in the first $i-1$ iterations, the function 
\texttt{Check}\_$\mathcal{P}(\mathbf{M}, \alpha, \tau_1, \tau_2)$ returns {\sc false}, 
and only at the $i$-th iteration does it return {\sc true}. 
Thus, the number of iterations follows a geometric distribution with a success probability 
of $p = (1 - \frac{1}{n^2})$. 
The expected value for such a geometric distribution is known to be 
$1/p = \frac{n^2}{n^2 - 1}$, which is $O(1)$.
\end{proof}

Now we can observe that if we want to construct codes guaranteeing flip resilience for $\alpha\in (0,1]$, 
it is enough to run Algorithm~\ref{alg:k_unknown} for $\alpha'=\alpha/2$. Since 
$\alpha'\in (0,1/2]$, Lemma~\ref{thm:unknown-k} implies the following.

%\begin{theorem}
%\label{thm:code-final}
%For 
%$\alpha\in(0,1]$, there is a Las Vegas algorithm that constructs 
%in expected polynomial time 
%an $(n,\alpha)$-URSC having \elongation\ 
%$\tau(n,k) = O\left(\frac{k^2}{\alpha^{2+\epsilon}}\ln n\right)$, for any $1 < k \le n$. 
%\end{theorem}

%\textcolor{red}{
\begin{theorem}
\label{thm:code-final}
For 
%any integer $n \geq 2$ and any real number 
$\alpha \in (0, 1]$, Algorithm \ref{alg:k_unknown} can be used to generate with high probability 
in polynomial time an $(n, \alpha)$-URSC having \elongation\
$\tau(n, k) = O\left(\frac{k^2}{\alpha^{2 + \epsilon}} \ln n\right)$,
for any $1 < k \leq n$. The same result can also be obtained in expectation.
\end{theorem}

\subsection{Ultra-resilient superimposed codes parameterized by length}

We could parameterize the ultra-resilient superimposed codes also by their length $t$, e.g., we could consider codes with length shorter than in Theorem~\ref{thm:code-final}.
%~\ref{thm:unknown-k}. 
More precisely, in Definition~\ref{def:shift_k_unknown-fuzzy}, we could consider any length $t$ of the code and vary the upper bound on the values of $k$ for which 
the elongation guarantee holds, which originally was $n$, to some parameter $\Delta$, where $\Delta$ is the maximum integer such that: 
$\Delta\le n$ and 
$\tau(n,\Delta)\le t$. This lowering of the upper bound on parameter $k$ is natural, because shorter codes could not guarantee a successful position of any element in every configuration of the codewords, due to existing lower bounds on the length of the code even if $k$ is known and there are no shifts, see e.g.,~\cite{AV1982}.

In general, if we cut the ultra-resilient superimposed code at some smaller length $t'<t$, the resulting code may not satisfy the elongation property in Definition~\ref{def:shift_k_unknown-fuzzy} for many values of parameter $k$, because the shifts in the original and the cut codes results in different configurations. However, our construction in Section~\ref{sec:unknown-k}, still works for shorter codes, with only few minor updates in the construction algorithm and in the analysis.
Mainly, we need to replace the formula for the length $t=(c/\alpha^{2+\epsilon}) n^2 \ln n$ by $t=(c/\alpha^{2+\epsilon}) \Delta^2 \ln n$, and consider only parameters $2\le k\le \Delta$. Note that throughout the analysis, all the probability bounds depend on the number of codewords $n$, the $\log n$ factor coming from the elongation function, and $\alpha$; all of these stay the same in the $t$-parameterized codes, hence the probabilities stay analogous (only constants may change, which we could accomodate here by taking a larger constant $c$).

This way we can prove the following extension of 
Theorem~\ref{thm:code-final} to codes with arbitrary length $t$.

\begin{theorem}
\label{thm:unknown-k-shorter-codes}
%Fix any $\epsilon > 0$. 
%Let $\alpha$ be a real number such that $0 < \alpha \leq 1$. 
%Given a sufficiently large integer $n$ and positive real constant $c$, 
Algorithm \ref{alg:k_unknown} can be modified to generate in polynomial time, with high probability, 
an $(n,\alpha)$-URSC of a given length $t$, with \elongation\ 
$\tau(n,k) = O\left(\frac{k^2}{\alpha^{2+\epsilon}} \ln n\right)$, for any $1\le k \le \Delta$, where~$\Delta$~is~the maximum integer satisfying: $\Delta\le n$ and $\tau(n,\Delta)\le t$. 
The same result can also be obtained in expectation.
\end{theorem}

\section{Code Applications
%Neighborhood learning and local broadcast 
in Uncoordinated Beeping Networks}
\label{sec:beeping}

\remove{%%%%%%%  START  REMOVE  %%%%%

The beeping model, introduced by Cornejo and Kuhn in their seminal work \cite{CornejoK10}, 
provides a simplified yet powerful framework for analyzing distributed algorithms under 
constrained communication environments. This minimalist communication model assumes that nodes 
interact over discrete time slots, where their only form of communication is through binary signaling: 
they can either emit a beep or remain silent. The feedback received by each node is limited 
to detecting whether there was a beep during that time slot, without any knowledge of which 
or how many nodes were responsible for the signal.
 
The model abstracts away many complexities of real wireless communication (e.g., signal strength, packet loss, etc.) 
and focuses on the fundamental coordination challenges arising from minimal information exchange.
A significant challenge in the beeping model arises from the fact that nodes do not know who beeped, 
and thus cannot directly exchange messages. They can only detect the presence or absence of a signal 
in their vicinity. Despite these constraints, the model captures essential features of low-power communication, 
which are relevant in the study of wireless sensor networks and IoT systems. 

}%%%%%%  END  REMOVE  %%%%%%%%%%%

\subsection{Model and problem}
\label{sec:model-beeping}
Let $G$ be an underlying beeping network with $n$ nodes, modeled as an undirected simple graph. 
Each node $v$ in the graph has a unique identifier from the set $\{1, \ldots, n\}$. Without loss of generality, 
we will refer to both the node and its identifier as $v$ . Each node is initially aware only of its own 
identifier, the total number of nodes $n$ , and an upper bound $\Delta$ on its degree (the number of neighbors).

In the \textit{uncoordinated setting}, nodes are activated in arbitrary rounds, as determined by a conceptual adversary.\footnote{Alternatively, the uncoordinated setting can be viewed as a temporal graph, where each node 
becomes ``visible" after its adversarial wake-up time, and each edge becomes ``visible'' 
during the earliest time interval when both of its endpoints~are~awake.}%
%\footnote{%
%Alternatively, we could look at graph $G$ in uncoordinated setting as an unknown temporal graph, in which each node is visible in the time suffix after its (adversarial) wake-up and each edge is visible in the minimal time interval of its two end nodes.} 
The goal for each node $v$ is to maintain a set of identifiers $N^*_v$ that satisfies the following~properties:

\vspace*{-1ex}
\begin{description}
    \item[\textbf{Inclusion:}] $N^*_v$ contains all the identifiers of the neighbors of $v$ in $G$.
\vspace*{-1ex}
    \item[\textbf{Safety:}] $N^*_v$  does not contain the identifier of any node that is not a neighbor 
    of $v$ in $G$.
\end{description}

\vspace*{-1ex}
This problem is referred to as \textit{neighborhood learning}.

%\footnote{%
%A node $v$ is required to learn and output the ids of its neighbors that are awaken not earlier than the node itself. The neighbors awaken earlier may stop beeping early enough so that  node $v$ does not have a chance to learn their ids. However, the problem does not exclude learning and outputting {\em some} of such neighbors by node $v$.}

In a more general problem, each node has to 
%learn and output 
maintain a set of input messages stored in its neighbors. 
%including neighbors who are awaken not earlier than the considered node (but may also contain input messages of other neighbor, if it can decode them despite of starting later then them) -- 
This problem is often called {\em local broadcast}.

{\em Time complexity} of a given problem in uncoordinated beeping networks is typically measured as a worst-case (over graph topologies and adversarial wake-up schedules) number of rounds from the moment when all nodes become awaken until the task is achieved, see~\cite{DBB2020}. In the case of neighborhood learning and local broadcast, the task is achieved if all locally stored sets of neighbors (resp., messages in neighbors) contain all neighbors (resp., all neighbors' messages). Note that, due to the Safety condition, when one of these two tasks is achieved, the locally stored sets of neighbors (resp., messages) do not change in the future rounds.
On the other hand, to guarantee the Inclusion property under arbitrarily long delays in adversarial wake-up schedule, algorithmic solutions have to be prepared for an arbitrarily long run -- therefore, periodic algorithms seem to be most practical and as such have been considered in the literature, see~\cite{DBB2020}.

It needs to be noted that the time complexity bound, denote it by $\mathcal{T}$, of our algorithms, proved in the analysis, satisfies even a stronger property: for any edge in the underlying network $G$, its both end nodes put each other (resp., each other's message) to the locally maintained set within time $\mathcal{T}$ after {\em both of them} become awaken. It means that we do not have to wait with time measurement until all nodes in the network are awaken, in case we are interested in specific point-to-point information exchange.
%\dk{TBA}

%\subsection{Deterministic local broadcast in uncoordinated beeping networks via ultra-resilient superimposed codes}
%\label{sec:beeping}

\subsection{Neighborhood learning in uncoordinated beeping networks}

Let $G$ be an underlying beeping network of $n$ nodes and node degree less than $\Delta$. 
Each node knows only its id and the integers $n$ and $\Delta$.~\footnote{The algorithm and its analysis work also if each node knows some polynomial upper bound on $n$ and a linear upper bound~on~$\Delta$.}
All nodes have the same $(n, \frac{3}{4})$-URSC $\M$ from Theorem~\ref{thm:unknown-k-shorter-codes}, of length set to $t=(c/\alpha^{2+\epsilon}) \Delta^2 \ln n$ and elongation $\tau(n,k) \le (c/\alpha^{2+\epsilon}) k^2 \ln n$, for some constant $c>0$, parameter $\alpha = \frac{1}{2} \cdot \frac{3}{4} = \frac{3}{8}$ and for any $k\le \Delta$ (the latter is because $\tau(n,\Delta)=t$ and $\Delta\le n$); in practice, it is enough that each node $v$ knows only the corresponding column $v$ of the code.\footnote{%
Constant $\frac{3}{4}$ is set arbitrarily -- it could be any constant in the interval $(\frac{1}{2},1)$ to allow 
the application of Theorem~\ref{thm:unknown-k}.}
Additionally, each node computes in the beginning its unique block ID of length 
$7+2\log n$, defined next.

%\B
%{\bf\em Block ID.}
\paragraph{Block ID.}
For a given node $v\in \{1,\ldots,n\}$, we define the {\em block ID of $v$}, denoted
$\bb_v$, as follows. It is a binary sequence of length $7+2\log n$, it starts with three 1s followed by three 0's and another 1.
To define the remaining $2\log n$ positions of the block ID, we take the binary representation of number $v$, of logarithmic length, and simultaneously replace each 0 by bits 01 and each 1 by bits 10. 

%\B
\paragraph{Main idea and intuitions.}

\gia{Here, we provide an overview of the key ideas and intuitions behind our approach, and we refer the reader 
to Appendix~D for the full proofs.
}
After awakening, a node periodically repeats a certain procedure -- see the pseudo-code Algorithm~\ref{alg:beeping}.
This periodicity is to assure that each node can properly pass its id, via sequence of beeps and idle rounds, to a neighbor who may be awaken at arbitrary time after the considered~node.
The main idea of the repeated part of Algorithm~\ref{alg:beeping}, see lines~\ref{alg:repeated-starts} - \ref{alg:repeated-ends}, which is in fact a form of another code, is as follows. A node $v$ uses its corresponding codeword $\bc_v$ from the ultra-resilient superimposed code $\M$ and substitutes each 1 in the codeword $\bc_v$ 
%at a node $v$ 
by its block ID $\bb_v$
%a specific (to each node) block of zeros and ones 
of length $7+2\log n$,
%%%%%$O(\log n)$, 
and each 0 in $\bc_v$ by the block of zeros of the same length.
%the former is called a {\em block ID of $v$}, and will be defined later. 
Then, each node beeps at rounds corresponding (according to its local clock) to positions with 1's in the obtained sequence, and stays idle 
%listening 
otherwise; see line~\ref{line:beeps} and preceding iterations in lines~\ref{line:main-for} and~\ref{line:second-for} corresponding to iterating over the length of the codewords $\bc_v,\bb_v$. 
The feedback, a beep or no beep, from the neighbors and the node itself is recorded in lines~\ref{line:records-1} and~\ref{line:records-0}.
Finally, a check is done (line~\ref{line:check}) whether the recorded feedback in the last $7+(2\log n)$ positions corresponds to any valid block ID, and if so, adding it to set $N^*_v$ (unless it is already there), see line~\ref{alg:repeated-ends}.

%At the end of this repeated code,
%%%%Algorithm~\ref{alg:beeping}, 
%each node analyzes the feedback (i.e., whether a beep was heard or not) during the rounds and identifies parts (corresponding to block IDs) that allow to identify~its~neighbors. 

Intuitively, the codewords from $\M$ are to assure that some beeped block ID of any neighbor node will not be overlapped by any block ID beeped by its competing neighbors at the same time -- therefore, it could guarantee the 
Inclusion property (see Lemma~\ref{lem:beeping-progress}).
More precisely, the shift property of the code guarantees that each neighbor has a unique 1 in its $\bc_v$, corresponding to some block ID of $v$, (i.e., while other neighbors have only 0's in its overlapping blocks), regardless of the adversarial shift of the code. However, such single 1 does not guarantee passing id to the neighbor -- this is why we need to substitute each 1 in the code by the block ID that allows decoding of the actual id if there are no beeping from other nodes at that time. This is challenging, because not only the original codewords from $\M$ could be adversarially shifted, but also the blocks themselves. Hence, to assure no beeping of other neighbors during the time when a block ID is beeped, we use 
{the isolation property}
%the slip property 
-- at some point, all neighbors not only have one overlapping block of 0's, but also the preceding and 
the next blocks are the block of 0's. 
(This is because 
{the isolation property}
%slip property 
guarantees that the preceding and the next position in the original codewords $\bc_w$ of other at most $\Delta-2$ neighbors have to be all 0's.)  

The reason why we cannot use just a node id as its block ID is the following.
The last $2\log n$ bits are to assure that any two aligned block IDs differ on at least two positions (which helps to assure that if a block ID is heard, it is not a ``beeping superposition'' of other block IDs but indeed a single node beeping its own block ID). The first $7$ bits are to guarantee that no genuine block ID could be heard while there is no beeping block ID aligned. 
This assures Safety property, see Lemma~\ref{lem:beeping-reliability} for details, and helps to fulfill Inclusion property too.

%The pseudo-code of our local broadcast algorithm is given in Algorithm~\ref{alg:beeping}.

\begin{algorithm}[ht!]
\small
\caption{\small Neighborhood learning algorithm in an uncoordinated beeping network; pseudo-code for node $v$
after its (uncoordinated) wake-up}\label{alg:beeping}
\linespread{1.1}\selectfont 
\textbf{Input:} Integers $n \ge \Delta \ge 1$, identifier $v\in \{1,\ldots,n\}$, a codeword $\bc_v$ 
from $(n, \frac{3}{4})$-URSC $\M$ from Theorem~\ref{thm:unknown-k-shorter-codes}, 
of length set to $t=(c/\alpha^{2+\epsilon}) \Delta^2 \ln n$ \\
\textbf{Maintained:} 
Set $N^*_v$ of identifiers
%of all neighbors of $v$ in the underlying unknown network~$G$.

\SetKwFunction{MyFunction}{Check\_$\mathcal{P}$}
\SetAlgoNlRelativeSize{0}
\SetAlgoNlRelativeSize{-1}

\BlankLine
   {$\bb_v \gets$ block ID of $v$} 
\;
   {$N^*_v \gets \emptyset$} \;
\While{{\tt True}}{
   {$\sigma_v \gets$ sequence of $t\cdot (7+2\log n)$ zeros\label{alg:repeated-starts}} 
\;
\For{$i=1$ to $t$\label{line:main-for}}{
    \For{$j=1$ to $7+2\log n$\label{line:second-for}}{
        \If {
        $\bc_v[i] \land \bb_v[j]$
        }
        {
             $v$ beeps\label{line:beeps}
             \tcp{beeping block ID of $v$ when $\bc_v[i]=1$}
        }
        \If{$v$ has beeped or heard a beep}{$\sigma_v[(i-1)\cdot (7+2\log n)+j] \gets 1$\label{line:records-1}} 
        \Else{%$v$ records 
        $\sigma_v[(i-1)\cdot (7+2\log n)+j] \gets 0$\label{line:records-0}}
\If{sub-sequence $\sigma_v[(i-2)\cdot (7+2\log n)+j+1 \mod t\cdot (7+2\log n),$ $ \ldots, (i-1)\cdot (7+2\log n)+j \mod t\cdot (7+2\log n)]$ is equal to block ID of some $w\in\{1,\ldots,n\}$\label{line:check}}{add $w$ to set $N^*_v$ (unless it is already there) \label{alg:repeated-ends}}
    }
}
}
%
%\For{$i=1$ to $(t-1)\cdot (7+2\log n)+1$}{
%    \If{sub-sequence $\sigma_v[i, \ldots, i+(7+2\log n)-1]$ is equal to block ID of some $w\in\{1,\ldots,n\}$}{add $w$ to set $N^*_v$ (unless it is already there) \label{alg:repeated-ends}} 
%}
%}
%\KwRet {Set $N_v$}
%\;

\BlankLine

\end{algorithm}

\begin{theorem}
\label{thm:beeping}
Algorithm~\ref{alg:beeping} can be instantiated with some $(n, \frac{3}{4})$-URSC $\M$ of length $t=O(\Delta^2\log n)$ and it guarantees learning  neighborhoods deterministically by each node in $O(\Delta^2\log^2 n)$ rounds after awakening of the node and the neighbors.
\end{theorem}

\subsection{From learning neighbors to local broadcast}

Now suppose that each node $v$ has a message of length $\cM$. It splits it into a sequence of $\log n$ messages of size $\cM/\log n$ each, say $M_v[1],\ldots,M_v[\cM/\log n]$.
Then, it creates extended messages $M^*_v[i]$, for any $1\le i\le \cM/\log n$, as follows: it puts a binary representation of $v$ by $\log n$ bits first, then it puts a single bit equal to $1$ for $i=1$ and to $0$ otherwise (this bit indicates whether it is the first extended message or not), and then appends $M_v[i]$. The length of $M^*_v[i]$ is $2\log n+1 \le O(\log n)$. 

There are two changes in Algorithm~\ref{alg:beeping}.
First, we now treat $M^*_v[i]$ as a set of new ids on node $v$.
Hence, we need a $(2n^2,3/4)$-URSC $\cM$ from Theorem~\ref{thm:unknown-k-shorter-codes}, but of asymptotically same length $t=O(\Delta^2\log n)$ as for neighborhood learning (as obviously $\log(2n^2)=\Theta(\log n)$ and we want to use the same bound $\Delta$).
Second, node $v$ in its $i$th periodic procedure (recall that such a procedure is in lines~\ref{alg:repeated-starts} -~\ref{alg:repeated-ends}) uses $M^*_v[i\mod \cM/\log n]$ as its id.
Let's denote this modified algorithm as the \textit{ultra-resilient Beeping Algorithm}.

The analysis is analogous, in particular, all nodes receive all messages $M^*_w[i]$ from their neighbors $w$, they are able to decode that they are from $w$ by looking at the first $\log n$ of the decoded $M^*_w[i]$, to identify the starting message by looking at bit $\log n+1$ (and then identifying the last message from $w$ by looking for the last bit 0 at that position before getting 1 at that position), and getting the actual content by looking at the last $\log n$ bits of $M^*_v[i]$ (and concatenating contents starting from the first identified message from $w$ up to the last piece). The only differences are that now both the code $\cM$ and block IDs are a constant factor longer, and node $v$ has to wait $\cM/\log n$ periodic procedures to be able to decode all parts of the~original~message.
Hence, the analog of Theorem~\ref{thm:beeping} can be proved for local broadcast:

\begin{theorem}
\label{thm:local-broadcast-beeping}
The \textit{Ultra-resilient Beeping Algorithm} 
%can be modified such that, instantiated with some $(2n^2, \frac{3}{4})$-ultra-resilient invariant superimposed code $\M$ of length $t=O(\Delta^2\log n)$, it 
guarantees deterministic local broadcasting of an input message of length $\cM$ by each node to each of its neighbors in $O(\Delta^2\log n \cdot (\cM+\log n))$ rounds after awakening of the node and the neighbor.
\end{theorem}

%\gdm{We should discuss the practical importance of having the codes for $k$ unknown rather than $k$ known. This is because there is already a construction for $k$ known. So we have to convince the reader that finding a construction for $k$ unknown is essential for a class of problems (the more the better) for which the assumption of $k$ know is not practical.}

\remove{

\paragraph{Communication in radio and beeping networks}

In non-adaptive communication in radio~\cite{ClementiMS01} and beeping~\cite{CornejoK10} networks, often classic superimposed codes are used. Typically, since the number of contending neighbors may not be known, classic superimposed codes for values $k=2,4,8,\ldots,n$ are used, intertwined and aligned. Using our ultra-resilient  superimposed codes (for unknown $k$) instead, allows to reduce the communication time by a logarithmic factor.

\subsection{Dynamic broadcast on MAC}

Parameter $\alpha$ allows pipelining several messages within one code execution.

\subsection{Beeping communication model}

\text{https://dl.acm.org/doi/pdf/10.1145/3382734.3405699}
Extended version here:
\text{https://hal.science/hal-02860827/document}
Section 3 is very relevant to our results

\subsection{Conflict-Avoiding codes for Optical Networks}

\text{https://epubs.siam.org/doi/10.1137/06067852X}

\subsection{Topology-transparent scheduling}

\text{https://link.springer.com/article/10.1007/s11276-006-6528-z}
}

%TBA -- some more about superimposed codes, but specially about their applications to distributed computing and state of the art about applications we consider:

%Contention resolution: TBA

%Dynamic broadcast: TBA

%Communication in the beeping model: TBA

%Optical codes: TBA

%\B
\section{Open Directions}

There are several promising directions for future research. First, expanding the applications of 
URSCs to additional domains, such as genomic alignment and dynamic database search, could offer 
substantial advantages due to their inherent fault-tolerant and asynchronous properties. 
Exploring applications in other distributed and parallel computing contexts, as well as investigating 
further properties, like code weight and fairness in mechanism design, are intriguing directions for 
advancing URSC capabilities.

Closing the small gap between the code length of URSCs and the theoretical lower bound is a challenging 
yet valuable endeavor. This, along with deeper exploration of ultra-resilient properties in new settings, 
holds potential for further enhancing the impact of URSCs.

In summary, our contributions position URSCs as a robust, scalable solution to foundational challenges 
in distributed computing. We anticipate that URSCs will stimulate future research into resilient 
coding for asynchronous and dynamic environments, pushing the boundaries of fault-tolerant coding theory.
More conclusions and future directions are deferred to Appendix~\ref{sec:final}.

\section*{Acknowledgments}
We thank Ugo Vaccaro for many inspiring and fruitful discussions.

\bibliographystyle{plain}
\bibliography{bibliography}

%\newpage
\appendix

\ \\

\begin{center}
{\bf\Large Appendix}
\end{center}

\section{Further Conclusions and Future Directions}
\label{sec:final}

In this work, we introduced Ultra-Resilient Superimposed Codes (URSCs), which extend 
classic superimposed codes with a stronger codewords' isolation property 
%than the classic counterpart
and unprecedented resilience to cyclic shifts and bitwise corruption, without requiring prior 
knowledge of the number of concurrent codewords, $k$. Our polynomial-time construction is the first efficient 
implementation of superimposed codes that can handle arbitrary shifts without assuming $k$, 
achieving near-optimal length that matches the best-known codes lacking these ultra-resilient properties.

URSCs provide significant improvements for several distributed 
%computing 
problems where synchronization, essential for classic codes, is prohibitively expensive. 
For instance, in uncoordinated beeping networks, URSCs achieve nearly two orders of magnitude improvement 
in local broadcast efficiency. They also support deterministic contention resolution in multi-access channels, underscoring their potential to improve robustness and adaptability in real-world distributed systems.

\B
\paragraph{Improvements in other distributed problems.}
Beyond the immediate applications to beeping networks and contention resolution, URSCs offer broader 
potential across diverse fields where synchronization constraints are challenging. For example, 
URSCs may improve scenarios where packets arrive dynamically at stations on a channel or in 
multi-hop wireless networks. In these cases, current efficient scheduling methods~\cite{CholviGJK22} 
rely on extended superimposed codes; our codes could replace them, potentially reducing their length to a near-optimal 
$O(k^2 \log n)$.

Additionally, our codes can benefit dynamic scheduling problems, such as topology-transparent 
and asynchronous schedules, 
where they efficiently manage activation times of participants in distributed systems. Prior studies, 
such as those by Chu, Colbourn, and Syrotiuk~\cite{CCSb2006, CCSa2006}, have highlighted the need 
for adaptable codes in such asynchronous contexts, and URSCs provide an efficient, near-optimal solution 
in these environments.
 
%\paragraph{Other potential applications.}
The broader applicability of URSCs extends to distributed database search, pattern recognition, 
and genomic data analysis. In distributed databases, traditional superimposed codes require aligned 
file descriptors, which can be impractical in asynchronous or adversarial retrieval systems. Here, 
URSCs overcome these limitations by tolerating out-of-order data arrivals and data corruption, thanks to 
their configurable parameter $\alpha$, making them ideal for robust data handling. 
This approach is particularly relevant in the context of data-dependent superimposed codes by
Indyk~\cite{I1997}, which have applications in pattern matching but lack tolerance to misalignment and faults.

URSCs also hold promise for genomic sequence alignment in large, distributed datasets, where their resilience to misalignments and adversarial conditions can enhance fault tolerance in biological data processing.

\section{Contention Resolution on a Multiple-Access Channel}
\label{sec:applications}

In this section, we 
introduce a generalized version of the Contention Resolution (CR) problem on a multiple-access channel and demonstrate how to apply our codes to get an efficient solution.
%The best existing deterministic protocol is based on an existential approach, providing transmission sequences that guarantee a latency of $O(k^2\log n)$ for each station arriving at an arbitrary time, where $k$ is the maximum number of stations joining the channel~\cite{DEMARCO20231}. 
%By utilizing our polynomially constructible codes, we now show that it is possible to efficiently solve the contention resolution problem with the same asymptotic latency.

\subsection{Model and problem}
%{Formal definitions of relevant distributed computing problems}
\label{sec:model-contention-resolution}

We consider a group of $n$ stations connected to a shared transmission medium called a \textit{shared channel}. 
Each station $v$ is uniquely identified by an ID in the range $\{1, \ldots, n\}$. 
For the sake of presentation and to simplify the notation, we treat $v$ and its ID interchangeably, 
meaning that $v$ not only represents the station but also serves as its ID, \textit{i.e.}, $v \in \{1, \ldots, n\}$.

At most $k$ stations out of the $n$, where $k$ is {\em unknown}, may become active, 
potentially in different time rounds.
The communication occurs in synchronous rounds, meaning that the clocks of all stations tick at the same rate. 
However, the most general assumption is without a \textit{global clock} and \textit{no system-based synchronization}. 
Each station measures time using its own local clock, which starts at the time round when it is 
activated, and is independent of other stations' clocks.
An instance of the Contention Resolution problem is represented as a 
pair $(T, \delta)$, where $T = \{v_1, v_2, \ldots, v_k\}$ is a set of $k$ stations, and $\delta: T \rightarrow \mathbb{N}$ 
is a function that maps each station $v_i$ to its activation time $R_i = \delta(v_i)$.
Without loss of generality 
$0 = \delta(v_1) \le \delta(v_2) \le \cdots \le \delta(v_k)$.

We consider {\em non-adaptive} deterministic algorithms for contention resolution (CR), 
where each station's actions are predefined at the start of execution, based solely on the parameter $n$ and the station's ID.
More precisely,
for each station $v$,  a \textit{transmission vector} $\cS_v$ consists of a sequence of bits that correspond to the time slots of $v$'s local clock. 
If the $r$-th bit of $\cS_v$ is 1, the station transmits in the $r$-th slot of its local time; otherwise, it remains silent.
A protocol $\cA$ for the Contention Resolution problem consists of a collection of $n$ transmission vectors, one for each station.

A protocol $\cA$ solves the Contention Resolution problem with \textit{latency} $\tau(\cA)$ if, for every possible instance $(T, \delta)$, 
it enables every station $v \in T$ to transmit successfully within $\tau(\cA)$ rounds from its activation time.

The classic version of the problem, as described above, aims at one successful transmission of each activated station. The generalized CR, studied in this work, assumes that each activated station $i$ has a number $s_i$ of packets to be successfully transmitted on the channel, i.e., it needs $s_i$ successful transmissions. Let $s$ be the upper bound on the values of $s_i$.

\begin{figure}[t!]
	\centering
\begin{tikzpicture}
	\matrix [nodes=draw,column sep=1mm]
	{
		\node[draw=none] {$v_1$}; & \node[draw=none]{$v_2$}; 
		& \node[draw=none] {$v_3$}; & \node[draw=none] {$v_4$}; 
		& \node[draw=none] {$v_5$}; & \node[draw=none] {$v_6$}; 
		&  \node[draw=none] {$\cdots$}; & \node[draw=none] {$v_n$}; \\
		\node {1}; & \node{0}; & \node {1}; & \node {1}; & \node {1}; & \node {1}; & ; & \node {0};\\
%		\node {1}; & \node{1}; & \node {1}; & \node {0}; & \node {1}; & \node {0}; & ;& \node {0};\\
%		\node {0}; & \node{1}; & \node {0}; & \node {0}; & \node {0}; & \node {0}; & ;& \node {1};\\
%		\node {0}; & \node{0}; & \node {1}; & \node {1}; & \node {1}; & \node {1}; &  ;& \node {0};\\
		\node {1}; & \node{0}; & \node {1}; & \node {0}; & \node {0}; & \node {0}; 
		&  \node[draw=none] {$\cdots$}; & \node {1};\\
		\node {0}; & \node{0}; & \node {1}; & \node {1}; & \node {0}; & \node {1}; & ;& \node {1};\\
%		\node {1}; & \node{1}; & \node {1}; & \node {0}; & \node {0}; & \node {0}; & ;& \node {0};\\
%		\node {0}; & \node{0}; & \node {0}; & \node {0}; & \node {1}; & \node {1}; & ;& \node {0};\\
		\node {0}; & \node{1}; & \node {0}; & \node {1}; & \node {1}; & \node {1}; & ;& \node {1};\\
        \node[draw=none] {}; & \node[draw=none]{}; 
       & \node[draw=none] {}; & \node[draw=none] {}; 
       & \node[draw=none, minimum height=20pt] {\smash{\raisebox{-40\depth}{$\vdots$}}}; & \node[draw=none] {}; 
       &  \node[draw=none] {}; & \node[draw=none] {}; \\
%		\node {0}; & \node{0}; & \node {1}; & \node {1}; & \node {0}; & \node {1}; & ;& \node {1};\\
%\node {0}; & \node{0}; & \node {1}; & \node {1}; & \node {0}; & \node {0}; & ;& \node {0};\\
\node {1}; & \node{0}; & \node {1}; & \node {1}; & \node {0}; & \node {0}; & ;& \node {0};\\
\node {1}; & \node{1}; & \node {0}; & \node {1}; & \node {1}; & \node {1}; & \node[draw=none] {$\cdots$};& \node {1};\\       			
	};
\end{tikzpicture}
\caption{\textit{The one-to-one correspondence between stations and columns of $\M$:
the transmission vector of station $v_i$ will be the column at index $v_i$. }}
\label{fig:M}	
\end{figure}

\subsection{Description of the protocol}
All stations are equipped with the same $t \times n$ binary matrix $\M$, which is a 
{$(n, \alpha)$-URSC, as specified in Definition~\ref{def:shift_k_unknown-fuzzy}, for an arbitrary constant $\alpha\in (0,1)$, say $\alpha=1/2$.
We also need the code to satisfy Definition~\ref{CBproperty}; hence, we could take our constructed code as in Theorem~\ref{thm:code-final}.}
%{def:shift_k_unknown}. 
We can establish a one-to-one correspondence between the set of stations and the columns of 
matrix $\M$ (see Figure \ref{fig:M}). 
Specifically, each station $v$ is associated with the 
column at index $v$ in $\M$ (recall that $v \in \{1,\ldots,n\}$). 

We define our protocol $\cA$ as follows.
The transmission vector of any station $v$ is represented by the column 
at index $v$, which we denote as $\M(v)$. 
In other words, for each station $v$, we set {
\[
\cS_v = \M(v)\cdot \left\lceil\frac{s}{(1-\alpha)\cdot weight_n}\right\rceil
\ , 
\]
where $weight_n$ is the minimum of $|M(v)_{\tau_1(n,n),\tau_2(n,n)}|$ over stations $v$ and $\M(v)\cdot x$ denotes concatenation of codeword $M(v)$ $x$ times. Recall that $|M(v)_{\tau_1(n,n),\tau_2(n,n)}|$ denotes the number of 1's in the interval $[\tau_1(n,n),\tau_2(n,n)]$ of codeword $M(v)$, where $\tau_1(n,n),\tau_2(n,n)$ are as in Definition~\ref{CBproperty}.}

\begin{example}
Suppose the stations are provided with the matrix depicted in Figure \ref{fig:M}. In this scenario, stations 
$v_1$, $v_3$ and $v_6$ have the following transmission vectors, respectively:
\[(1, 1, 0, 0, \ldots,  1, 1) \ , \ \ 
(1, 1, 1, 0,  \ldots, 1, 0) \ , \ \ 
(1, 0, 1, 1, \ldots, 0, 1) \ .\]
Consider an instance where 
$\delta(v_1) = 0$,
$\delta(v_3) = 2$ and
$\delta(v_6) = 3$. 
The $r$-th local round of station $v_3$ is synchronized with
round $r + \delta(v_3) - \delta(v_1)  = r + 2$ of $v_1$'s local clock
and with round $r + \delta(v_3) - \delta(v_6) = r - 1$ of 
$v_6$'s local clock (see Figure \ref{fig:Trans}).
\end{example}

In general, we can state the following fact.

\begin{fact}\label{f:synch}
For any two stations $v$ and $v'$, the $i$-th round of station $v$ is synchronized with round $i' = i + \delta(v) - \delta(v')$ of station $v'$. 
In particular, if $i'$ is negative, it indicates that $v'$ has not been activated yet.
\end{fact}

\begin{figure}[t!]
	\centering
\begin{tikzpicture}
	\matrix [nodes=draw, row sep  = 2mm, column sep = 2mm]
	{
		\node[draw=none] {Communication rounds:}; & \node[draw=none]{1}; & \node[draw=none]{2}; & \node[draw=none]{3}; & \node[draw=none]{4}; 
%  & \node[draw=none]{5}; & \node[draw=none]{6}; & \node[draw=none]{7}; & \node[draw=none]{8}; & \node[draw=none]{9};  
		& \node[draw=none]{$\cdots$}; \\		
		\node[draw=none] {$v_1$}; & \node{1}; & \node{1}; & \node{0}; & \node{0}; 
%  & \node{1}; & \node{0}; & \node{1}; & \node{0}; & \node{0};  
& \node[draw=none]{$\cdots$};\\
		\node[draw=none] {$v_3$}; & \node{}; & \node{};  & \node{1}; & \node{1}; 
%  & \node{0}; & \node{1}; & \node{1}; & \node{1}; & \node{1}; 
& \node[draw=none]{$\cdots$};\\
		\node[draw=none] {$v_6$}; 
%  & \node{}; & \node{}; 
& \node{}; & \node{}; & \node{};  & \node{1}; 
%& \node{0}; & \node{0}; & \node{1};  
& \node[draw=none]{$\cdots$};\\
	};
\end{tikzpicture}
\caption{\textit{The `out of sync' of transmission vectors for stations $v_1$, $v_3$, and $v_5$ following an instance where $\delta(v_1) = 0$, $\delta(v_3) = 2$, and $\delta(v_6) = 3$.
The empty squares indicate times when the station was not yet activated.}}
\label{fig:Trans}
\end{figure}

\subsection{Correctness and complexity}

Let us fix an arbitrary instance $(T, \delta)$ of the contention resolution problem.
{Let 
%$s^*=c'\min\{s,(1-\alpha)\cdot weight_n\}$ and $k^* = \sqrt{s^*/\log n}$, 
$k^*=k+c'\cdot \left\lceil\frac{s}{\log n}\right\rceil$
for some suitable constant $c'$ to be specified later.}
We aim to show that every station $v \in T$ will successfully transmit {at least $s$ times} within $t$ rounds after activation, 
where $t$ is {
$\tau(\cA)=\Theta\Big(\big(k+\frac{s}{\log n}\big)^2\log n\Big)$. More precisely, if $k^*\le n$ then $\tau(\cA)$ is the elongation $\tau(k^*)$ of code $M$ for parameter $k^*$, and otherwise $\tau(\cA)$ is the length of $M$ multiplied by $\left\lceil\frac{s}{(1-\alpha)\cdot weight_n}\right\rceil
$ (this is by observing that $\frac{k^*}{n}\ge \left\lceil\frac{s}{(1-\alpha)\cdot weight_n}\right\rceil$ for sufficiently large constant $c'$ in the definition of $k^*$). Let us also generalize definition of $weight_n$ to $weight_{\ell}$, for any $\ell\le n$, to be equal to the minimum number of 1's in any codeword of $M$ in the interval $[\tau_1(n,\ell),\tau_2(n,\ell)]$ of positions. We will prove later in Lemma~\ref{lepa1} that in our constructed code, $weight_{\ell}>\frac{3}{5}\sqrt{c}k\ln n = \Omega(k\log n)$.\footnote{{Formally, Lemma~\ref{lepa1} proves the lower bound with high probability, but this probability is enough to carry on through the derandomization argument and also easy to check by the algorithm in linear time per codeword.}}}
%the number of rows in $\M$. 
%In other words, we have $\tau(\cA) \leq t$. 
We start with the following~lemma.

\begin{lemma}\label{l:collision}
	Let $(T, \delta)$ be an arbitrary instance of contention resolution problem, and consider a station $v \in T$ that is transmitting during the $i$-th round of its local clock.
	Suppose there exists another station $v' \in T\setminus \{v\}$ transmitting simultaneously.
	Let $\mathbf{c} = \mathbf{M}(v)$ and $\mathbf{y} = \mathbf{M}(v')$. There exists a vector $\mathbf{z} \in S(\mathbf{y})$ such that both $\mathbf{c}$ and $\mathbf{z}$ have a 1 in their $i$-th position.
\end{lemma}

\begin{proof}
Let's explicitly represent the bit components of the 2 transmission vectors as $\mathbf{c} = (c_1, c_2, \ldots, c_t)$ and $\mathbf{y} = (y_1, y_2, \ldots, y_t)$.
Let $d = \delta(v) - \delta(v')$.
As a consequence of Fact \ref{f:synch}, the $i$-th bit component
of $\mathbf{c}$
is processed at the same computation round as bit $y_{i+ d}$ of
$\mathbf{y}$.
Since the two stations transmit simultaneously in that computation round, we have $c_i = 1$ and $y_{i+d} = 1$. 

It's worth noting that the value of $d$ can be positive, zero, or negative, depending on whether $v'$ is activated before, simultaneously with, or after $v$, respectively.
Therefore, $y_{i+d} = 1$ implies $y_{i \oplus (t+d)} = 1$, where we need to recall that $\oplus$ denotes addition modulo $t$. 
Hence, if we set
$\mathbf{z} = (y_{1 \oplus (t+d)}, y_{2 \oplus (t+d)}, \ldots, y_{t \oplus (t+d)})$,
then both $\mathbf{c}$ and $\mathbf{z}$ have a 1 in their $i$-th position, and the lemma follows.
\end{proof}

\begin{theorem}
\label{thm:codes-vs-CR}
Let $(T, \delta)$ be an arbitrary instance of the generalized contention resolution problem. 
Protocol $\cA$ solves the generalized contention resolution problem for any $k\le n$ contenders (with $k$ unknown) and guarantees each of them at least $s$ successful transmissions 
with latency {$\tau(\cA) \leq t = O((k+\frac{s}{\log n})^2\log n)$.}
\end{theorem}

\begin{proof}
{Note first that if $k^*>n$ then we consider $\lceil k^*/n \rceil$ concatenations of $M$, and in each concatenation it is enough to prove the number of submission is $weight_n=\Omega(n\log n)$ (the asymptotics is by Lemma~\ref{lepa1}). In other words, this corresponds to case $k^*=n$ that occurs $\lceil k^*/n \rceil$ times in the code. Therefore, in the remainder it is sufficient to focus on case $k^*\le n$.}

Let $v$ be an arbitrary station in $T$. 
Starting from computation time $\delta(v)$, station $v$ transmits according to its transmission vector $\cS_{v} = \M(v) = (c_1, \ldots, c_t)$: 
for $j = 1,2, \ldots, t$ it transmits at time 
$\delta(v)+ j$ if and only if  $c_j = 1$. 

Let 
%$r = |\bc|$
{$r = |[\bc]_{[\tau_1(n,k^*),\tau_2(n,k^*)]}|$} and $\{j_1, j_2, \ldots, j_r\}$ be the collection of indices from the transmission vector $(c_1, \ldots, c_t)$ where the value is 1. 
Consequently, for station $v$, the transmission will happen 
exclusively during its local rounds ${j_1}, {j_2}, \ldots, {j_r}$.
{Note that $r\ge weight_{k^*}$, by the definition of $weight_{k^*}$, and by Lemma~\ref{lepa1}, we get $r\ge \frac{3}{5}\sqrt{c}(k+c'\frac{s}{\log n})\log n > 2k\log n+2s$ for sufficiently large constants $c,c'$. It follows that 
%$r-s>2k\log n$
$r/2>s$, and consequently, $r-s>r/2$.}

Suppose, by contradiction, that 
%none 
{less than $s$}
of these $r$ transmissions 
%is 
{are}
successful. In this scenario, it implies that for 
%each 
{more than $r-s$ of}
$v$'s 
local 
{times} 
$j_i$, for $i=1,2,\ldots,r$,
there is (at least) another station $v^{j_i} \in T \setminus\{v\}$,
associated with $v$'s local time $j_i$, that transmits simultaneously at $v$'s local time $j_i$.

Let $\bc = \M(v)$ and $\by_i = \M(v^{j_i})$
for $i = 1, 2, \ldots, r$.
Applying Lemma \ref{l:collision} repeatedly  
for $i = 1, 2, \ldots, r$, we get that for each 
$i = 1, 2, \ldots, r$,  there exists a
vector $\bz_i \in S(\by_i)$ such that both 
$\bc$ and $\bz_i$ have a 1 in their $i$-th position.
Consequently, the bitwise OR of all these
vectors, that is 
$\bz = \bigvee_{i=1}^r \bz_i$, must have weight
{$|\bz| > r -s$. By applying the previously proved bound $r-s>r/2$, we get $|\bz| > r/2 = r\alpha$.}
%= |\bc|$. 
{Hence, we have found a vector $\bz \in S_{\vee}(T \setminus \{ \bc \})$
 such that $|\bc \wedge \bz| > |\bc|-s>\alpha |\bc|$.
This contradicts the fact that $\M$ is an ultra-resilient 
superimposed code satisfying also Definition~\ref{CBproperty}, for which it is required, instead, that
for all $\bz \in S_{\vee}(T \setminus \{ \bc \})$,
$|\bc_{[\tau_1(n,k^*),\tau_2(n,k^*)]} \wedge \bz_{[\tau_1(n,k^*),\tau_2(n,k^*)]}| < \alpha |\bc_{[\tau_1(n,k^*),\tau_2(n,k^*)]}|$, where $0< \alpha \le 1$. The last inequality comes from applying the Collision Weight Inequality (Equation~(\ref{conjineq2}) in Definition~\ref{CBproperty}) $k-1$ times -- specifically, to the $k-1$ superimposed codewords of the contending stations.}
\end{proof}

\subsection{Other related work on contention resolution} \label{sec:other-related}

Deterministic contention resolution has been studied in a slightly easier setting where all 
$k$ stations arrive at the same time. In this scenario, the latency is significantly lower, 
$\Theta(k\log n)$, with the upper bound established in \cite{KG1985} and a construction 
given in \cite{K2005}, while the lower bound is proved in \cite{CMS2003}.
%\textcolor{red}
{The dynamic case, assuming the availability of a global clock, was examined in \cite{M2015}.}
 
Recently, fairness in contention resolution has emerged as a challenging objective~\cite{ChionasCKK23}. 
In our work, we consider the setting with dynamic arrivals 
%\textcolor{red}
{and without global clock,} where another optimization 
criterion is the number of transmissions (corresponding to the weight of the codewords). 
Our approach also shows slight improvement in this aspect, cf. \cite{M2021}.

\section{Proof of Lemma 2}\label{proofs}

In proving Lemma \ref{lep}, we proceed by separately bounding the probabilities of satisfying 
each of the two inequalities of the Collision Bound Property, and then we combine these bounds 
in the final proof. This is done in Appendix \ref{WIsub} and \ref{CWIsub}, respectively.

%\textcolor{red}
{Some preliminary mathematical tools that will be pivotal in the proofs of this section are
deferred to the Appendix \ref{prel}. 
This includes the \textit{rearrangement inequality} \cite{HLP1934}, a fundamental 
result in combinatorial mathematics essential for bounding and optimizing sums involving products 
of sequences of probabilities, as well as specific results on bounding summations through integral approximation.
}

From this point onward, we assume the hypotheses of the lemma. 
Specifically, in addition to $n$ and $\alpha$, we fix $\epsilon > 0$ and set 
$\tau_1(n,k) = \frac{c}{64 } k^2 \ln n$ 
and $\tau_2(n,k) = \frac{c}{\alpha^{2+\epsilon}} k^2 \ln n$, where $c > 0$ is a sufficiently 
large constant. Moreover, we consider $\mathbf{M} = \mathcal{M}(n, \alpha, \epsilon, c)$ to 
be any matrix generated according to the random construction of Definition \ref{randomatrix}, 
and $1 < k \leq n$ is any given (unknown) parameter.

\subsection{Weight Inequality}\label{WIsub}

We begin with inequality (\ref{conjineq}), which compares the weight 
of the upper segment to that of the lower segment of a column relative to $\tau_1$ and $\tau_2$.
Specifically, our goal here is to prove the following lemma.
\begin{lemma}\label{lep1}
The probability that the Weight Inequality (\ref{conjineq}) 
does not hold is at most 
$\frac{2}{c^2} \cdot n^{-8\ln(4/\alpha)} $.
\end{lemma}

We need some preliminary result.
Let $Y^{\top}_j$ and $Y^{\perp}_j$ be the random variables denoting the weight of the  
subcolumns $\mathbf{[c_j]}_{[0, \tau_1(n,k)]}$ and 
$\mathbf{[c_j]}_{[\tau_1(n,k), \tau_2(n,k)]}$, respectively. 
\begin{lemma}\label{Eup}
We have
\[
 \frac{1}{8}  \sqrt{c} k \ln n          
         < E[Y^{\top}_j] < 
         \frac{1}{2}  \sqrt{c} k \ln n
 \]
\end{lemma}
\begin{proof}
We have 
\[
	E[Y^{\top}_j]   =     \sum_{b=0}^{\tau_1(n,k)} p(r) 
	                =    \ln n \sum_{b=0}^{(c/ 64) k^2} \frac{1}{\sqrt{b+1}} .
\]
We can find upper and lower bounds for $\frac{1}{\sqrt{b+1}}$ as follows.
\[
\frac{\left( \sqrt{b+2} + \sqrt{b+1} \right)
	\left( \sqrt{b+2} - \sqrt{b+1} \right)}{\sqrt{b+1}} 
 = \frac{1}{\sqrt{b+1}} =
\frac{\left( \sqrt{b+1} + \sqrt{b} \right)\left( \sqrt{b+1} - \sqrt{b} \right)}{\sqrt{b+1}}.
\]
From which we get
\begin{equation}\label{ineqs}
     2\left( \sqrt{b+2} - \sqrt{b+1} \right) < \frac{1}{\sqrt{b+1}} < 2\left( \sqrt{b+1} - \sqrt{b} \right).
\end{equation}

Hence,
\begin{eqnarray*}          
2\ln n \left( \sqrt{(c/64)  k^2 + 2} - 1 \right)  
        & < E[Y^{\top}_j] <
        & 2  \ln n \left( \sqrt{(c/64)  k^2+1} \right), \\
\ln n \left( \sqrt{(c/64) k^2} \right)	      
        & < E[Y^{\top}_j] < 
        & 4  \ln n \left( \sqrt{(c/64) k^2} \right), \\            
\frac{1}{8}  \sqrt{c} k \ln n          
        & < E[Y^{\top}_j] < 
        & \frac{1}{2}  \sqrt{c} k \ln n \ .
\end{eqnarray*}

\end{proof}

Analogously, we can now estimate the expected value of $Y^{\perp}_j$.
\begin{lemma}\label{Edown}
We have
\[
  \left (\frac{1}{\alpha^{1+\epsilon/2}} - \frac{1}{8} \right )\sqrt{c} k \ln n  
      < E[Y^{\perp}_j] < 
   4 \left (\frac{1}{\alpha^{1 + \epsilon/2}}  \right )\sqrt{c} k \ln n .
\]
\end{lemma}
\begin{proof}
We proceed as in the proof of Lemma \ref{Eup}. We have
\[
	E[Y^{\perp}_j] 
                =    \sum_{r=\tau_1(n,k)}^{\tau_2(n,k)} p(r) \nonumber\\
	            =    \ln n \sum_{b= (c /64)  k^2}^{(c/\alpha^{2 + \epsilon}) k^2} \frac{1}{\sqrt{b +1}}.
\]
Using the inequalities (\ref{ineqs}), we have

\begin{eqnarray*}
  \ln n \sum_{b= (c /64) k^2}^{(c/\alpha^{2 + \epsilon}) k^2} 2\left( \sqrt{b+2} - \sqrt{b+1} \right)  
  & < E[Y^{\perp}_j] < &
  \ln n \sum_{b= (c /64) k^2}^{(c/\alpha^{2 + \epsilon}) k^2} 2 \left( \sqrt{b+1} - \sqrt{b} \right), \\
  2 \ln n \left( \sqrt{(c/\alpha^{2 + \epsilon}) k^2 + 2} - \sqrt{(c /64)  k^2+1} \right)
   & < E[Y^{\perp}_j] < &
2  \ln n \left( \sqrt{(c/\alpha^{2 + \epsilon}) k^2 + 1} - \sqrt{(c /64)  k^2} \right), \\
\ln n \left( \sqrt{(c/\alpha^{2 + \epsilon}) k^2 } - \sqrt{(c / 64)  k^2} \right)
   & < E[Y^{\perp}_j] < &
4\ln n \sqrt{(c/\alpha^{2 + \epsilon}) k^2 }, \\  
\ln n \left( \frac{\sqrt{c}\,k}{\alpha^{1 + \epsilon/2} } - \frac{\sqrt{c}\, k}{8 } \right)
   & < E[Y^{\perp}_j] < &
  4\ln n \frac{\sqrt{c}\,k}{\alpha^{1 + \epsilon/2} }, \\
  \left (\frac{1}{\alpha^{1+\epsilon/2}} - \frac{1}{8} \right )\sqrt{c} k \ln n  
     & < E[Y^{\perp}_j] < &
   4 \left (\frac{1}{\alpha^{1 + \epsilon/2}}  \right )\sqrt{c} k \ln n .
\end{eqnarray*}

\end{proof}

Since the random construction of the matrix (Definition \ref{randomatrix}) sets each entry of the columns to 1 or 0
independently, both $E[Y^{\top}_j]$ and $E[Y^{\perp}_j]$ can be considered as sums of independent random variables. 
Consequently, we can apply the Chernoff bound to evaluate the deviation of $Y^{\top}_j$ and $Y^{\perp}_j$
respectively above the expected value $E[Y^{\top}_j]$ and below the expected value $E[Y^{\perp}_j]$.

\begin{lemma}\label{lepa1}
We have,
\begin{equation*}
    	\Pr\left(Y^{\top}_j >  \frac{3}{5}\sqrt{c}\, k \ln n \right) 
     <  \frac{1}{c^2} \cdot n^{-8\ln(4/\alpha)}  \ , 
\end{equation*}
\begin{equation*}
    	\Pr\left(Y^{\perp}_j <  
     \frac{7}{10} \left (\frac{1}{\alpha^{1+\epsilon/2}} - \frac{1}{8} \right ) \sqrt{c}\, k \ln n \right)  
     <  \frac{1}{c^2} \cdot n^{-8\ln(4/\alpha)}  \ .
\end{equation*}
\end{lemma}

\begin{proof}

Letting $\delta = 1/5$, and recalling the upper and lower bounds for 
$E[Y^{\top}_j]$ established in Lemma \ref{Eup}, we have:
    \begin{eqnarray*}
    	\Pr\left(Y^{\top}_j >  \frac{3}{5} \sqrt{c}\, k \ln n \right)   
    	&\le &  \Pr\left(Y^{\top}_{j} \ge \frac{6}{5}\cdot E[Y^{\top}_j] \right) \\ 
    	&=   &  \Pr\left(Y^{\top}_{j} \ge (1+\delta) E[Y^{\top}_j] \right) \\      
    	&\le & \exp\left(-\frac{\delta^2 E[Y^{\top}_j]}{3}\right) \;\;\; \text{(see, for example, Eq. (4.2) in \cite{MitUpf})} \nonumber \\
    	&\le & \exp\left(- \frac{(1/5)^2  \sqrt{c} k \ln n}{24 }\right)  \\
    	&\le & \exp\left(- \frac{(1/5)^2  \sqrt{c} \ln (1/\alpha) \ln n}{24 }\right)  
              \;\; \text{(recalling that $e^{-k} < \alpha \le 1$}) \\  
    	&\le & \exp\left(- \frac{(1/5)^2  \sqrt{c} \left( \frac{\ln (4/\alpha)}{2}-\frac{\ln 4}{2} \right) \ln n}{24 }\right) \\             
    	&<& \frac{1}{c^2} \cdot n^{-8\ln(4/\alpha)} \ , 
    \end{eqnarray*}
   where the last inequality holds for a sufficiently large constant $c$.
Analogously, we can now let $\delta = 3/10$ and use the lower bound for 
$E[Y^{\perp}_j]$ established in Lemma \ref{Edown}.
    \begin{eqnarray*}
    	\Pr\left(Y^{\perp}_j <  \frac{7}{10}\cdot 
         \left( \frac{1}{\alpha^{1 + \epsilon/2}} - \frac{1}{8}\right )\sqrt{c} k \ln n \right)   
    	&\le &  \Pr\left(Y^{\perp}_{j} \le \frac{7}{10}\cdot E[Y^{\perp}_j]\right) \nonumber\\ 
    	& =  &  \Pr\left(Y^{\perp}_{j} \le (1 - \delta)\; E[Y^{\perp}_j]\right) \nonumber\\      
    	&\le & \exp\left(-\frac{\delta^2  
         \left (\frac{1}{\alpha} - \frac{1}{8}\right )\sqrt{c} k \ln n}{2}\right) \\ 
        & & \text{(see, for example, Eq. (4.5) in \cite{MitUpf})} \nonumber \\
    	&\le & \exp\left(-\frac{\delta^2  
        \sqrt{c} \ln(1/\alpha) \ln n}{16}\right) \\ 
        & & \text{(recalling that $e^{-k} < \alpha \le 1$}) \\
    	&\le & \exp\left(-\frac{(3/10)^2  
         \sqrt{c} \left( \frac{\ln(4/\alpha)}{2} -\ln 2 \right)\ln n}{16}\right) \\         
    	&<& \frac{1}{c^2} \cdot n^{-{8 \ln(4/\alpha) } } \ , \nonumber
    \end{eqnarray*}
    where the last inequality holds for a sufficiently large constant $c$.   
\end{proof}

We can now return to the Weight Inequality and prove Lemma \ref{lep1}.

\medskip
\noindent
\textbf{Proof of Lemma \ref{lep1}}.
Let $w_1 = \frac{3 }{5}\sqrt{c}\, k \ln n$ and 
$w_2 = \frac{7}{10}\left(\frac{1}{\alpha^{1+\epsilon/2}} - \frac{1}{8} \right) \sqrt{c} k \ln n$.
Recalling that $\alpha \le 1$, we can observe that 
\begin{eqnarray*}
    w_1 & =    & \frac{3}{5}\, \sqrt{c}\, k \ln n \\
        & <    & \frac{7}{10} \frac{7}{ 8} \, \sqrt{c}\, k \ln n \\
        & \le  & \alpha \cdot\frac{7}{10} 
        \left(\frac{1}{\alpha } - \frac{1}{8 } \right)   \, \sqrt{c}\, k \ln n \\        
        & \le  &  \alpha\cdot \frac{7}{10}\left(\frac{1}{\alpha^{1+\epsilon/2}} - \frac{1}{8} \right) \sqrt{c} k \ln n\\
        & =    & \alpha\, w_2 \ .
\end{eqnarray*}

In view of this, the inequality (\ref{conjineq}) is guaranteed to be satisfied 
if the following events occur simultaneously: $(Y^{\top}_j \le w_1)$ and $(Y^{\perp}_j \ge w_2)$.
Hence,  by applying Lemma \ref{lepa1}, the inequality is not satisfied 
with a probability of at most
\begin{equation*}
    \Pr(Y^{\top}_j > w_1) + \Pr(Y^{\perp}_j < w_2) < 
    \frac{2}{c^2} \cdot n^{-8\ln(4/\alpha)}  
    \ .
\end{equation*} 

\qed

\subsection{Collision Weight Inequality}\label{CWIsub}
Now it is the turn of the Collision Weight Inequality.
For $d \in \{-1,0,1\}$ and $0 \le i \le t-1$, let $X^d_{j,j'}(i)$ be the random variable denoting 
$|\mathbf{[c_j]} \land \mathbf{[c_{j'}]}(i + d + t)|$.
The left-hand side of inequality (\ref{conjineq2}) is a random variable 
$X_{j,j'}(i)$ such that 
\begin{equation*}%\label{eq:d}
X_{j,j'}(i) \le X^{-1}_{j,j'}(i) + X^{0}_{j,j'}(i) + X^{1}_{j,j'}(i) 
\ .
\end{equation*}
Our goal is to show that the two random variables $X_{j,j'}(i)$ and $Y_j^{\perp}$ satisfy the
Collision Weight Inequality. The first step is to determine an upper bound on the expectation of 
$X{j,j'}(i)$. We can observe that
\[
      E[X_{j,j'}(i)] \le \sum_{d=-1,0,1} E[X^{d}_{j,j'}(i)] .
\]
For $d \in \{-1,0,1\}$, we have
\[
    	E[X^{d}_{j,j'}(i)] = \sum_{r=\tau_1(n,k)}^{\tau_2(n,k)} p(r)\; p\left((r+i+d)\; \text{mod } t\right) .
\]
The value of $E[X^{d}_{j,j'}(i)]$ is significantly influenced by the size of $p\left((r+i+d) \mod t\right)$, 
which can vary considerably within the interval $[\tau_1(n,k), \tau_2(n,k)]$ depending on the shift 
$0 \le i < t$.  
In this context, the following three cases can occur:

\begin{itemize}
    \item \textbf{Case 1:} For all $r$ in the interval $\tau_1(n,k) \leq r \leq \tau_2(n,k)$, we have 
$p((r+i+d)\; \text{mod } t) > \frac{\alpha^{1+\epsilon/2} }{\sqrt{c}\; k} \ln({4}/{\alpha})$.
%\dk
{In other words, all probabilities of the shifted column are large within the considered positions.}
    \item \textbf{Case 2:} For $\lambda = \frac{\sqrt{c}}{\ln({4}/{\alpha})}, 
    \frac{\sqrt{c}}{\ln({4}/{\alpha})}+1 ,\ldots, 
    \frac{\sqrt{c} n}{2 k}$,
     and all $r$ in the interval $\tau_1(n,k) \leq r \leq \tau_2(n,k)$, we have 
  $\frac{\alpha^{1 + \epsilon/2} }{2\lambda k} \le p((r+i+d) \; \text{mod } t) 
  \le \frac{\alpha^{1+ \epsilon/2}  }{ \lambda k}$.
  %\dk
  {In other words, all probabilities of the shifted column are relatively small within the considered positions, and thus, from the definition of blocks, which are longer than $\tau_2(n,k)-\tau_1(n,k)$ for the considered probability, all probabilities are within a constant factor from each other.}

    \item \textbf{Case 3:} There is some $\tau'\in (\tau_1(n,k),\tau_2(n,k)]$ such that: for all $r$ in the interval $\tau_1(n,k) \leq r \leq \tau'$, we have 
$p((r+i+d)\; \text{mod } t) < \frac{\alpha^{1+\epsilon/2} }{2\sqrt{c}\; n}$, and for all $\tau' \le r \le \tau_2(n,k)$ we have $p((r+i+d)\; \text{mod } t) > \frac{\alpha^{1+\epsilon/2} }{\sqrt{c}\; k} \ln({4}/{\alpha})$.
%\dk
{In other words, all probabilities of the shifted column are large starting from some intermediate position $\tau'$, and smallest possible before that.}

\end{itemize}

The following lemmas establish upper and lower bounds on the expectation of $X^d_{j,j'}(i)$ 
%\dk
{for each of the fist two cases; case 3 is a simple combination of case 2 in some prefix of $[\tau_1(n,k),\tau_2(n,k)]$ and case 1 in the remaining suffix.}

\begin{lemma}\label{case1}
    In case 1, we have
    \[
    \frac{\ln n \ln({4}/{\alpha})}{2} 
    < E[X^d_{j,j'}(i)] <
    \ln n \left( 3 + \ln\left(\frac{64}{\alpha^{2 +\epsilon }}  \right) \right).
    \]
\end{lemma}
\begin{proof}
Let us begin with the lower bound. We have,
\begin{eqnarray*}
   E[X^d_{j,j'}(i)] 
          & > & \sum_{r= \tau_1(n,k) }^{\tau_2(n,k)} p(r)\; \frac{\alpha^{1+\epsilon/2} }{\sqrt{c}\; k}\ln({4}/{\alpha})  \\
          & = & \frac{\alpha^{1+\epsilon/2} }{\sqrt{c}\; k}\ln({4}/{\alpha}) \sum_{r= \tau_1(n,k) }^{\tau_2(n,k)} p(r)  \\
          & > & \frac{\alpha^{1+\epsilon/2} }{\sqrt{c}\; k}\ln({4}/{\alpha})  \cdot E[Y^{\perp}_j]   \\           
          & > & \frac{\alpha^{1+\epsilon/2} }{\sqrt{c}\; k}\ln({4}/{\alpha})  
                 \left (\frac{1}{\alpha^{1 + \epsilon/2}} - \frac{1}{8} \right )\sqrt{c} k \ln n 
          \;\; \text{ (by Lemma \ref{Edown})}\\
          & > & \frac{ \ln n \ln({4}/{\alpha})}{2 }.
\end{eqnarray*} 
The upper bound is more complex because in this case we need to consider all possible shifts of the second 
column, $j'$. To address this, we will use the rearrangement inequality.

Let $P^{i+d} = \{ p((r+i+d)\; \text{mod } t) \mid \tau_1(n,k) \leq r \leq \tau_2(n,k) \}$, and 
let $y^{i+d}_{\tau_1(n,k)}, \ldots, y^{i+d}_{\tau_2(n,k)}$ 
denote the elements of $P^{i+d}$ arranged in decreasing order. Formally, 
for $d \in \{-1,0,1\}$, we have defined a set
$\{y^{i+d}_{\tau_1(n,k)}, \ldots , y^{i+d}_{\tau_2(n,k)}\} = P^{i+d}$,
where $y^{i+d}_{\tau_1(n,k)} \ge \cdots \ge y^{i+d}_{\tau_2(n,k)}$.
By the rearrangement inequality (see Theorem \ref{th:hardy}) we have that 
\begin{eqnarray*}
	\sum_{r=\tau_1(n,k)}^{\tau_2(n,k)} p(r)\; p\left((r+i+d)\; \text{mod } t\right) 
	& \le & \sum_{r=\tau_1(n,k)}^{\tau_2(n,k)} p(r)\; y^{i+d}_r 
 \ .
\end{eqnarray*}
Moreover, since for $\tau_1(n,k) \le r \le \tau_2(n,k)$, 
$y^d_{r} \le p(r - \tau_1(n,k))$, then for $d = -1, 0, 1$:
\begin{eqnarray*}
	E[X^d_{j,j'}(i)] 
 	& \le & \sum_{r=\tau_1(n,k)}^{\tau_2(n,k)} p(r)\; y^{i+d}_r \\
	& \le &  \sum_{r=\tau_1(n,k)}^{\tau_2(n,k)} p(r)\; p(r-\tau_1(n,k)) \\
	& =  & \ln n \sum_{b = (c / 64  ) k^2 }^{(c/\alpha^{2 + \epsilon}) k^2} 
	\frac{1}{\sqrt{b + 1}} \cdot \frac{1}{\sqrt{b - (c / 64 ) k^2 + 1}} \\
	& =  & \ln n \sum_{b = (c / 64 ) k^2 +1}^{(c/\alpha^{2 + \epsilon} ) k^2 +1} 
	\frac{1}{\sqrt{b }} \cdot \frac{1}{\sqrt{b - (c / 64 ) k^2 }} \\ 
	& <  &  \ln n \left( 3 + \ln\left(\frac{c k^2}{\alpha^{2 + \epsilon}} +1 \right) 
            - \ln\left(  \frac{c k^2}{64  } + 1 \right) \right) 
               \,\,\,\,\,\,\,\, \text{ (by Lemma \ref{sum})}\\           
    & < &   \ln n \left( 3 + \ln\left(\frac{64}{\alpha^{2 +\epsilon }}  \right) \right)    
 \ .
\end{eqnarray*}

\end{proof}

\begin{lemma}\label{case2}
In case 2, we have that for $\lambda = \frac{\sqrt{c}}{\ln({4}/{\alpha})}, 
    \frac{\sqrt{c}}{\ln({4}/{\alpha})}+1 ,\ldots, 
    \frac{\sqrt{c} n}{2 k}$:
    \[
    \frac{\sqrt{c} \ln n }{4\lambda}
    < E[X^d_{j,j'}(i)] <
     \frac{2 \sqrt{c} \ln n }{\lambda }.
    \]
\end{lemma}
\begin{proof}
In this case, we have,
\begin{eqnarray*}
E[X^d_{j,j'}(i)] 
          & \ge & \sum_{r= \tau_1(n,k) }^{\tau_2(n,k)} p(r)\; \frac{\alpha^{1+\epsilon/2} }{2\lambda k} \\
          & = & \frac{\alpha^{1+\epsilon/2} }{2\lambda k} \sum_{r= \tau_1(n,k) }^{\tau_2(n,k)} p(r)  \\
          & = & \frac{\alpha^{1+\epsilon/2} }{2\lambda k}\cdot E[Y^{\perp}_j]   \\           
          & > & \frac{\alpha^{1+\epsilon/2} }{2\lambda k} \left (\frac{1}{\alpha^{1 + \epsilon/2}} - \frac{1}{8} \right )\sqrt{c} k \ln n 
          \;\; \text{ (by Lemma \ref{Edown})}\\
          & > & \frac{\sqrt{c} \ln n }{4\lambda }.
\end{eqnarray*}
Analogously,
\begin{eqnarray*}
   E[X^d_{j,j'}(i)] 
          & \le & \sum_{r= \tau_1(n,k) }^{\tau_2(n,k)} p(r)\; \frac{\alpha^{1+\epsilon/2} }{\lambda k} \\
          & = & \frac{\alpha^{1+\epsilon/2} }{\lambda k} \sum_{r= \tau_1(n,k) }^{\tau_2(n,k)} p(r)  \\
          & = & \frac{\alpha^{1+\epsilon/2} }{2\lambda k}\cdot E[Y^{\perp}_j]  \\
          & < & \frac{\alpha^{1+\epsilon/2} }{2\lambda k}\cdot 4\left (\frac{1}{\alpha^{1 + \epsilon/2}}  \right )\sqrt{c} k \ln n\\
          & = & \frac{2 \sqrt{c} \ln n }{\lambda }.
\end{eqnarray*}  
\end{proof}

The following lemma guarantees that the Collision Weight Inequality will be satisfied with high probability.

\begin{lemma}\label{CWIup}
The probability that the Collision Weight Inequality
does not hold is at most 
$$\frac{4}{c^2}\cdot n^{-8\ln(4/\alpha)} .$$
\end{lemma}
\begin{proof}
  Let $w = \frac{7}{10}\left(\frac{1}{\alpha^{1+\epsilon/2}} - \frac{1}{8} \right) \sqrt{c} k \ln n$.
The Collision Weight Inequality is guaranteed to be satisfied if the following two events occur
simultaneously:
\[
       X_{j,j'}(i) \le \left \lfloor \frac{\alpha\, w - 1}{k-1} \right \rfloor,
\]
\[
      \left \lfloor \frac{\alpha\, w - 1}{k-1} \right \rfloor 
      \le \left\lfloor \frac{\alpha  \; Y^{\perp}_j  - 1}{k - 1} \right\rfloor. 
\]
Hence, it is not satisfied with probability at most 
\begin{eqnarray}
 & &\Pr\left(X_{j,j'}(i) > \left \lfloor \frac{\alpha\, w - 1}{k-1} \right \rfloor \right) + 
    \Pr\left(\left \lfloor \frac{\alpha\, w - 1}{k-1} \right \rfloor 
    < \left\lfloor \frac{\alpha  \; Y^{\perp}_j  - 1}{k - 1} \right\rfloor\right) \nonumber \\
     & < & 
%    3\Pr\left(X^d_{j,j'}(i) > \frac{1}{3}\left \lfloor \frac{\alpha\, w - 1}{k-1} \right \rfloor \right) +
%    \frac{1}{c} \cdot n^{-8\ln(4/\alpha)} \nonumber \\
%    & < & 
    3 \Pr\left(X^d_{j,j'}(i) > \frac{1}{3} \left \lfloor \frac{\alpha\, w - 1}{k-1} \right \rfloor \right) +
    \frac{1}{c^2} \cdot n^{-8\ln(4/\alpha)}, \label{add}
\end{eqnarray}
where the last inequality follows by the second inequality of Lemma \ref{lepa1}. 

To complete the proof, it remains to analyze the probability that 
$X^d_{j,j'}(i) > \frac{1}{3} \left \lfloor \frac{\alpha\, w - 1}{k-1} \right \rfloor$.

Substituting $w$ to the right-hand-side of this inequality, we have
\begin{eqnarray}
\frac{1}{3}\left \lfloor \frac{\alpha\, w - 1}{k-1} \right \rfloor 
    & > & 
    \frac{1}{3} \cdot \alpha \left( \frac{7}{10} \left( \frac{1}{\alpha^{1+\epsilon/2}} - \frac{1}{8} \right) \sqrt{c} \ln n \right) \nonumber \\
    & = &
    \frac{1}{3} \frac{7}{10} \left( \frac{1}{\alpha^{\epsilon/2}} - \frac{\alpha}{8} \right) \sqrt{c} \ln n \nonumber \\
    & > &
    \frac{1}{10} \left( \frac{1}{\alpha^{\epsilon/2}}  \right) \sqrt{c} \ln n. \label{llas} 
%    & > &
%     \frac{1}{140} (1 + \sqrt{c}) \ln n \left(2 + \ln\left(\frac{64}{\alpha^{2+\epsilon }}  \right) \right), \label{llas}  
\end{eqnarray}
%where the last step holds for any $\epsilon > 0$.
%and a sufficiently large constant $c$.

Because each entry of a column is independently set to either 1 or 0 in the random construction of the matrix, 
%the expected value $E[X^d_{j,j'}(i)]$ 
$X^d_{j,j'}(i)$ for $d = -1, 0, 1$ can be viewed as a sum of independent random 
variables. Therefore, we can utilize the Chernoff bound to assess the probability that $X^d_{j,j'}(i)$ 
exceeds its expected value. 
We need to analyze each of the three cases defined at the beginning of subsection \ref{CWIsub}.

\medskip
\textbf{Case 1.}
Let $m = (1 + \delta)  \ln n \left(3 + \ln\left(\frac{64}{\alpha^{2+\epsilon }}  \right) \right)$, with
$\delta = \ln c$. It holds that
\[
   m \le \frac{1}{10} \left( \frac{1}{\alpha^{\epsilon/2}}  \right) \sqrt{c} \ln n
   < \frac{1}{3}\left \lfloor \frac{\alpha\, w - 1}{k-1} \right \rfloor \text{ (by (\ref{llas}))}.
\]
Applying the upper bound of Lemma \ref{case1}, we also have
\[
    (1 + \delta) E[X^d_{j,j'}(i)] < m \le \frac{1}{3}\left \lfloor \frac{\alpha\, w - 1}{k-1} \right \rfloor .
\]
Therefore,
\begin{eqnarray*}
    \Pr\left(X^d_{j,j'}(i) > \frac{1}{3}\left \lfloor \frac{\alpha\, w - 1}{k-1} \right \rfloor \right)
    & \le &
	\Pr\left(X^d_{j,j'}(i) \ge m \right) \\
	& = & \Pr\left(X^d_{j,j'}(i)
	\ge (1+\delta)  E[X^d_{j,j'}(i)] \right)   \\
    & \le & \left ( \frac{e^{\delta}}{(1+\delta)^{1+\delta}} \right )^{E[X^d_{j,j'}(i)]} \;\;\; 
    \text{(see, for example, Eq. (4.1) in \cite{MitUpf})} \\ 
    & \le & \exp \left (-\frac{\delta^2 E[X^d_{j,j'}(i)]}{2 + \delta} \right ) \;\;\; 
    \left(\text{because $\ln\left(1+\delta \right) \ge \frac{2\delta}{2+\delta}$}\right) \\
    & \le & \exp \left (-\frac{\ln c\, E[X^d_{j,j'}(i)]}{2} \right ) \\ 
    & \le & \exp \left (-\frac{\ln c \ln n \ln (4/\alpha)}{4} \right ) \;\;\; \text{(by the lower bound of Lemma \ref{case1})}\\      
	& = & \frac{1}{c^2}\cdot n^{-8\ln(4/\alpha)} 
 \ ,
\end{eqnarray*}
where the last inequalities hold for sufficiently large $c$.
Plugging this bound into (\ref{add}) the lemma is proved for case~1.

\medskip
\textbf{Case 2.}
Let $\lambda$ be any value in $\left\{ \frac{\sqrt{c}}{\ln({4}/{\alpha})}, \frac{\sqrt{c}}{\ln({4}/{\alpha})} +1 ,\ldots, \frac{\sqrt{c} n}{2 k} \right\}$.
%    \[
%    \frac{\sqrt{c} \ln n }{4\lambda }
%    < E[X^d_{j,j'}(i)] <
%     \frac{2 \sqrt{c} \ln n }{\lambda }.
%   \]
Let $m = (1 + \delta) \frac{2 \sqrt{c} \ln n}{\lambda}$, where $\delta = \ln{c}$. 
We can observe that for any $\epsilon > 0$ the following inequality holds
for every value of $\lambda$ in 
$\left\{ \frac{\sqrt{c}}{\ln({4}/{\alpha})}, \frac{\sqrt{c}}{\ln({4}/{\alpha})} +1 ,\ldots, \frac{\sqrt{c} n}{2 k} \right\}$:
\[
   m \le \frac{1}{10} \left( \frac{1}{\alpha^{\epsilon/2}}  \right) \sqrt{c} \ln n.
\]
Applying the upper bound of Lemma \ref{case2} and (\ref{llas}) 
respectively on the left-hand-side and the right-hand-side
of this inequality, we also get
\[
 (1+\delta) E[X^d_{j,j'}(i)] < m \le \frac{1}{3}\left \lfloor \frac{\alpha\, w - 1}{k-1} \right \rfloor .
\]
Therefore,
\begin{eqnarray*}
    \Pr\left(X^d_{j,j'}(i)  >  \frac{1}{3}\left \lfloor \frac{\alpha\, w - 1}{k-1} \right \rfloor \right) 
    & \le & \Pr(X^d_{j,j'}(i) > m) \\
    & \le & \Pr(X^d_{j,j'}(i) > (1 + \delta) E[X^d_{j,j'}(i)] ) \\
    & \le & \left ( \frac{e^{\delta}}{(1+\delta)^{1+\delta}} \right )^{E[X^d_{j,j'}(i)]} \;\;\; 
    \text{(see, for example, Eq. (4.1) in \cite{MitUpf})} \\     
    &\le & \exp\left(-\frac{\delta^2 E[X^d_{j,j'}(i)]}{2 + \delta}\right) 
    \;\;\; \left(\text{because $\ln\left(1+\delta \right) \ge \frac{2\delta}{2+\delta}$}\right)  \\
    &\le & \exp\left(-\frac{\ln{c}}{2} \frac{\sqrt{c}\ln n}{4\lambda} \right) 
    \;\;\; \text{(by the lower bound of Lemma \ref{case2})} \\   
    &\le & \exp\left(-\frac{\ln{c}\ln n \ln(4/\alpha)}{8} \right) \\
    & \le &  \frac{1}{c^2} \cdot n^{-8\ln(4/\alpha)},
\end{eqnarray*}
where the last steps hold for sufficiently large $c$. Plugging this bound into (\ref{add}) the lemma
is proved for case~2.

\medskip
\textbf{Case 3.}
%\dk
{It follows directly from applying the above analysis of case 2 for $\lambda=\frac{\sqrt{c}n}{2k}$ to the interval $[\tau_1(n,k),\tau')$ and the analysis of case 1 to the interval $[\tau',\tau_2(n,k)]$.
Then, we take the union bound of the two events.}
\end{proof}

Now, we are ready to complete the proof of Lemma \ref{lep}.

\medskip
\noindent
\textbf{Proof of Lemma \ref{lep}}.
The lemma follows by applying the union bound to the results of Lemmas~\ref{lep1}~and~\ref{CWIup}. 
%
%\dk{??? More details ???}
%
\qed

\subsection{Auxiliary results}
\label{prel}

The primary advantage of the Collision Bound Property is its ability to reduce 
a property concerning subsets of $k$ matrix columns to one involving only pairs of columns. 
The following result from combinatorial mathematics plays a crucial role in 
computing the expected number of positions where two arbitrary columns both contain a 1 (collision). 
Specifically, it helps to address the challenge posed by arbitrary shifts of the columns during 
these computations. In our context, the sequences of real numbers 
$x_1 \le \cdots \le x_n$ and $y_1 \le \cdots \le y_n$ represent the probabilities of encountering 
a 1 at corresponding positions in two columns. The rearrangement inequality allows us to handle 
the probability of collisions independently of the misalignment caused by arbitrary shifts between 
the two sequences.

\begin{theorem}[Theorem 368 (Rearrangement Inequality, \cite{HLP1934})]\label{th:hardy}
	For any choice of real numbers
	\[
	x_1 \le \cdots \le x_n \;\;\text{ and }\;\;  y_1 \le \cdots \le y_n
	\]
	and every permutation $y_{\sigma(1)} \le \cdots \le y_{\sigma(n)}$ 
	of $y_1 \le \cdots \le y_n$,
	\begin{equation}\label{rearrange}
	x_1 y_n + \cdots + x_n y_1 
	\le x_1 y_{\sigma(1)} + \cdots + x_1 y_{\sigma(n)} 
	\le x_1 y_1 + \cdots + x_n y_n.	    
	\end{equation}
\end{theorem}

The analysis of column-wise collisions is further facilitated by Lemma \ref{sum}, 
which aims to bound the sum of collisions between two shifted columns within specific rows 
of the matrix. 
The following lemma is preparatory to Lemma \ref{sum}.

\begin{lemma}\label{integral}
For $x > a$, it holds that
    \[
    \int \frac{1}{\sqrt{x(x-a)}}\, dx = 2 \ln\left|  \sqrt{x - a} + \sqrt{x} \right| + C.
    \]
\end{lemma}
\begin{proof}
     Letting $ u = \sqrt{x} $, we get $x = u^2$  and  $dx = 2u \, du$.
Substituting these into the integral, we get:
\begin{equation}\label{intfa}
\int \frac{1}{\sqrt{x(x-a)}} \, dx = \int \frac{1}{u} \cdot \frac{1}{\sqrt{u^2 - a}} \cdot 2u \, du
 = 2 \int  \frac{1}{\sqrt{u^2 - a}} \, du .    
\end{equation}
This is a standard integral (see \textit{e.g.} entry 2.271 (4) in \cite{gradshteyn2007}), 
and it is known that:
\[
\int \frac{1}{\sqrt{u^2 - a}} \, du = \ln\left| \sqrt{u^2 - a} + u   \right| + C.
\]
Substituting back $ u = \sqrt{x}$ and plugging the result 
into equation (\ref{intfa}), we obtain the lemma.

\end{proof}

\begin{lemma}\label{sum}
For $1 < h_1 < h_2$, it holds that
     \[
     \sum_{i = h_1 + 1}^{h_2 + 1} \frac{1}{\sqrt{i(i- h_1)}} 
     < 3 + \ln\left( {h_2+1} \right) - \ln\left({h_1+1}\right) .
     \]
\end{lemma}

\begin{proof}
    For $x > h_1$, the function $f(x) = \frac{1}{\sqrt{x(x -h_1)}}$ is monotonic decreasing. 
    Hence, we can apply the Riemann sum approximation as follows:
    \begin{eqnarray*}\label{eq:split}
    \sum_{i=h_1 + 1}^{h_2 + 1} \frac{1}{\sqrt{i(i - h_1)}} 
      & =   & \frac{1}{\sqrt{(h_1+1)}} + \sum_{i=h_1+2}^{h_2+1} \frac{1}{\sqrt{i(i-h_1)}}  \\
      & <   & 1 + \sum_{i=h_1 + 2}^{h_2 + 1} \frac{1}{\sqrt{i(i- h_1)}} \\
      & \le &  1 + \int_{h_1 + 1}^{h_2 + 1} \frac{1}{\sqrt{x(x- h_1)}}\, dx \\     
      & =   & 1 + \left[ 2\ln\left(\sqrt{x - h_1} + \sqrt{x} \right) \right]_{h_1 +1}^{h_2+1}
              \,\,\,\,\,\,\,\, \text{(by Lemma \ref{integral})}\\
      & =   & 1+ 2\ln \left(\sqrt{h_2 + 1 -h_1} + \sqrt{h_2 + 1}\right) - 
                   2\ln \left(1 + \sqrt{h_1+1}  \right)  \\
      & <   & 1  + 2\ln\left(2\sqrt{h_2+1}\right)-2\ln\left(\sqrt{h_1+1}\right) \\
      & =   & 1  + 2 \ln 2 + \ln\left( {h_2+1} \right) - \ln\left({h_1+1}\right) .
    \end{eqnarray*}    
\end{proof}

\section{Analysis of the Neighborhood-Learning Algorithm (Algorithm~\ref{alg:beeping}) and the Proof of Theorem~\ref{thm:beeping}}\label{app:neigh}

\begin{lemma}[Safety]
\label{lem:beeping-reliability}
If a node $v$ adds node $w$ to set $N^*_v$, then $w$ is a neighbor of $v$ in the underlying network~$G$.
%awaken not later than $v$.
\end{lemma}

\begin{proof}
The proof is by contradiction -- suppose that in some network and adversarial wake-up schedule, $v$ is not a neighbor of $w$ but $v$ puts $w$ to its set $N^*_v$. It can happen only if $v$ heard at some point a binary sub-sequence equal to the block ID of node $w$, denoted $\bb_w$.
To focus our attention, we denote that particular sub-sequence of $\sigma_v$ by $\gamma$; it is identical to $\bb_w$, and it occurs in particular time. 
Since $w$ is not a neighbor of~$v$, this heard sequence $\gamma$ comes from activities (beeping or not beeping) of actual neighbors of~$v$~after~awakening.

Let $w_1,\ldots,w_\ell$ be the minimal, in the sense of inclusion, subset of neighbors of $v$ that contribute to the considered sequence $\gamma$
%block ID of $w$ 
heard by node $v$, for some integer $\ell\ge 1$. More precisely, each bit 1 in $\gamma$ comes from beep of some (at least one) of the nodes $w_1,\ldots,w_\ell$ and none of them beeps in round corresponding to bit 0 in $\gamma$. 
By minimality of set $w_1,\ldots,w_\ell$, each of them has some of its block IDs overlapping sequence $\gamma$, as otherwise it would only have blocks of zeros overlapping and we could have removed it from set $w_1,\ldots,w_\ell$.
%
%Since each of $w_1,\ldots,w_\ell$ contributes to sequence $\gamma$ heard in some $7+2\log n$ consecutive rounds by $v$, each of them has to have its block ID overlapping (in time) with the rounds when $\gamma$ is heard by $v$ (otherwise, it would not contribute a single bit to $\gamma$). 

%Next, observe that 
\begin{claim}
\label{claim:alligned}
Each such overlapping block ID of nodes $w_1,\ldots,w_\ell$ must actually be aligned (in time) with $\gamma$, unless it finishes before round $7$ of $\gamma$ or starts in the last round of $\gamma$.     
\end{claim}

\noindent
{\bf\em Proof of Claim~\ref{claim:alligned}:}
Suppose, to the contrary, that some overlapping block ID of some $w_i$, for $1\le i\le \ell$, 
%starts before or after the first round of $\gamma$ while 
does not finish before round $7$ of $\gamma$ and does not start at the very last round of $\gamma$. 
We consider two complementary cases and obtain contradiction in each of them.

%In the former case, 
If this block ID starts before round $1$ of $\gamma$,
we observe that during the rounds $4,5,6$ of $\gamma$, node $w_i$ has to beep at least once. This is because these rounds align with rounds $j,j+1,j+2$ of the block ID of $w_i$, for some $j\ge 5$, and in every three of such rounds there is at least one beep (by the definition of the block ID).
We obtain contradiction with the fact that there is no beep in $\gamma$, which value is the same as the block ID $\bb_w$ of $w$, in the considered rounds $4,5,6$ (see again the definition of block ID).

In the remaining case, i.e., when the considered block ID of $w_i$ overlapping $\gamma$ starts after round $1$ and before the last round of $\gamma$, the first three consecutive 1's of that block would cause a beep at some position of $\gamma$ that stores $0$. This is because one of the following arguments, and obviously yields contradiction: (i) in the subsequence of $\gamma$, starting after position $1$, there are no three consecutive 1's, or (ii) if the block actually starts in the second-last round of $\gamma$, in the last two positions of $\gamma$ there is at least one 0.
$\qed$

\smallskip
%%Thus we proved that each of $w_1,\ldots,w_\ell$ is either aligned with $\gamma$ or finishes before position $7$ of $\gamma$ or starts at the last position of $\gamma$. 
It follows from Claim~\ref{claim:alligned} that at least one of the overlapping block IDs of $w_1,\ldots,w_\ell$ must be aligned with $\gamma$,
because otherwise none of them would cause the beeps in some positions $7,\ldots,7+2\log n -1$ of~$\gamma$. 
W.l.o.g. we could assume that the block ID of $w_1$ is aligned with $\gamma$. Recall that $w_1\ne w$, because $w_1$ is a neighbor of $v$ and $w$ is not. Hence, there is a position $j$ in the ids of $w_1$ and $w$ on which they differ. It implies that their block IDs are different at positions $7+2(j-1) +1$ and $7+2(j-1)+2$. In particular, in one of these position there is 0 in $\gamma$ and 1 in the considered block ID of $w_1$, which contradicts the definition of $\gamma$ and the set $w_1,\ldots,w_\ell$; more precisely, by definition, if $\gamma$ stores 0 at some position then all  $w_1,\ldots,w_\ell$ did not beep at the corresponding round. 
%as the beeps heard from neighbors of $v$, and thus also $w_1$, during the considered time period.  
\end{proof}

%\begin{lemma}[S\&S codes with no cyclic shifts]
%\label{lem:prefix-code}
%Consider code $\M$, which is an $(n,\Delta,\tau)$-ultra-resilient superimposed code from Theorem~\ref{}.
%Consider ??????
%\end{lemma}

\begin{lemma}[Inclusion]
\label{lem:beeping-progress}
If a node $w$ is a neighbor of node $v$, 
%awaken not earlier than $v$, 
then the block ID of $w$ can be found in some $\sigma_v[i, \ldots, i+(7+2\log n)-1]$, for some $1\le i\le (t-1)\cdot (7+2\log n)+1$, checked by $v$ at most $2t\cdot (7+2\log n)$ rounds after both $v$ and $w$ are awaken.
\end{lemma}

\begin{proof}
Consider a node $v$ and its neighbors $w_1,\ldots,w_\ell$, for some integer $0\le \ell< \Delta$, whose beeping overlap the execution of the algorithm by node $v$. 
Consider one of the neighbors of $v$, w.l.o.g. $w_1$. Consider the first round when both $v$ and $w_1$ are already awaken and $w_1$ starts its periodic repetition (line~\ref{alg:repeated-starts}). 
In what follows, for the sake of simplifying notation, we will be numbering rounds, starting from the considered round, $1,2,\ldots$.
(Note that such round occurs less than $t\cdot (7+2\log n)$ rounds after both these nodes are awaken.)

We define checkpoints at rounds $1+(i-1)\cdot (7+2\log n)$, for $i=1,\ldots,t$. (Recall that $t$ is 
the length 
of the ultra-resilient superimposed code $\M$.)
For every node $v,w_1,\ldots,w_\ell$, at each checkpoint $i$ we define the bit $i$ of its sequence to be equal to 1 if the node executes its block ID, and 0 otherwise (i.e., if it beeps its block of zeros). Observe that the resulting sequence is a shifted codeword of that node, and the sequence for $w_1$ is its original codeword $\bc_v$ (i.e., not shifted, because we assumed that it starts its periodic procedure at round we number $1$). 
Here by ``shifted'' we consider both: shifted cyclically, as in Definition~\ref{def:shift_k_unknown-fuzzy}, or ``shifted not cyclically'' in which we pad positions before the shift value with zeros.
The latter corresponds to the case when a neighbor has been awaken during the considered interval. Note that if the Definition~\ref{def:shift_k_unknown-fuzzy} holds for a configuration of codewords in some set $T$, shifted cyclically, it also holds if instead of cyclic shifts we pad positions before the shift values with zeros instead. This is because it may only create more zeros in the definition of vector $\bz^*$, which is even better from perspective of the guaranteed equation in Definition~\ref{def:shift_k_unknown-fuzzy}.

By 
%Definition~\ref{def:shift_k_unknown-fuzzy} of the ultra-resilient superimposed code $\M$, 
the fact that code $\M$ from Theorem~\ref{thm:unknown-k-shorter-codes} satisfies Definition~\ref{def:shift_k_unknown-fuzzy} of $(n,3/4)$-URSC with elongation guarantees for any $k\le \Delta$ (this is because $\tau(n,\Delta)=t$ and $\Delta\le n$), there is a position $i$ such that:
the codeword $\bc_{w_1}$ has $1$ at this position, while the other considered shifted codewords of $v,w_2,\ldots,w_\ell$ have zeros at position $i$ and $i+1$.
{Here we use the isolation property together with shift resilience.}
%(\textcolor{purple}{Here we use both shift resilience and slip resilience. 
%\textbf{We should again get rid of the slip property and mention the isolation property again.}})
It means that in the interval $[1+(i-1)\cdot (7+2\log n),i\cdot (7+2\log n)]$ node $w_1$ beeps its block ID, while nodes $v,w_2,\ldots,w_\ell$ stay idle. The latter holds because by the definition of checkpoint $i$, taken in round $1+(i-1)\cdot (7+2\log n)$, all these nodes have their $\bc_v[i]=\bc_{w_2}[i]=\ldots=0$ and thus execute their blocks of zeros in that round; even if they finish their blocks early after that round, because they all have also $\bc_v[i+1]=\bc_{w_2}[i+1]=\ldots=0$, they start executing another block of zeros (each of length $7+2\log n$), hence none of them ever beeps in the considered interval $[1+(i-1)\cdot (7+2\log n),i\cdot (7+2\log n)]$.

Since we assumed that $v$ is alive in the considered rounds, it will record the beeps heard in the interval $[1+(i-1)\cdot (7+2\log n),i\cdot (7+2\log n)]$. As we argued above, only its neighbor $w_1$ is beeping (its own block ID) in this period, among the neighbors of $v$ and $v$ itself, therefore $\sigma_v$ in this interval will be equal to the block ID of $w_1$. Hence, $w_1$ will be added to set $N^*_v$ in the end of the considered interval (line~\ref{line:check}). %nearest future update of set $N^*_v$ by $v$.
This happened at most $t\cdot (7+2\log n)$ rounds after the considered round we numbered $1$, because it takes place within the considered periodic procedure (lines~\ref{alg:repeated-starts} -~\ref{alg:repeated-ends}) of lenth $t\cdot (7+2\log n)$. 
%%that is, when $v$ finishes its periodic procedure containing the considered interval $[1+(i-1)\cdot (7+2\log n),i\cdot (7+2\log n)]$. 
As argued earlier, the considered round, from which we start counting in this analysis, occurs less than $t\cdot (7+2\log n)$ rounds after both $v$ and $w_1$ are awaken; hence, in less than $2t\cdot (7+2\log n)$ rounds after both $v$ and $w_1$ are awaken, $v$ updates its set $N^*_v$ by adding $w$, which block ID it has heard meanwhile.  
\end{proof}

\begin{proof}[Proof of Theorem~\ref{thm:beeping}]
The existence of code $\M$ used in Algorithm~\ref{alg:beeping}, constructable in polynomial time, is guaranteed by Theorem~\ref{thm:unknown-k-shorter-codes} in Section~\ref{sec:unknown-k}.

Safety property comes from Lemma~\ref{lem:beeping-reliability}, while Inclusion property -- from Lemma~\ref{lem:beeping-progress}.
Note that sets $N^*_v$ are monotonic, in the sense of inclusion -- nothing is ever removed from them. 
%Correctness -- that each node learns its exact neighborhood in $G$ during the algorithm -- comes directly from Lemmas~\ref{lem:beeping-reliability} and~\ref{lem:beeping-progress}.
Time complexity also follows from Lemma~\ref{lem:beeping-progress}, even in the stronger sense that each node learns its neighbor at most $2t\cdot (7+2\log n) \le O(\Delta^2\log^2 n)$ after both of them are awaken.\footnote{There is no need to wait with time measurement until other neighbors are awaken, as typically assumed in the literature; however, that classic measurement also holds in our case, because once a neighbor is added to set $N^*_v$, it says there and could be considered as already added after the round when all neighbors of $v$ become awake.}
%follows from combining the length of the code $t=O(\Delta^2\log n)$ with the $O(\log n)$ length of each block ID (substituting zeroes and ones in the code $\M$). 
\end{proof}

\end{document}